\documentclass[useAMS,usenatbib,referee]{biom}

\usepackage{fancyhdr}
\usepackage[toc, page]{appendix}
\usepackage{extramarks}
\usepackage{amsmath}
\usepackage{amsfonts}
\usepackage{tikz}
\usepackage{enumerate}
\usepackage[colorlinks]{hyperref}
\usepackage{amssymb}
\usepackage{mathtools}
\usepackage{dsfont}
\usepackage{textcomp}
\usepackage{setspace}
\usepackage{float}
\usepackage[strings]{underscore}
\usepackage{natbib}
\usepackage{threeparttable}
\usepackage{adjustbox}
\usepackage{array}
\usepackage{soul}
\usepackage{graphics}

\usetikzlibrary{shapes,shapes.geometric,automata,positioning,decorations}

\tikzset{
  semi/.style={
    semicircle,
    draw
  }
}

\usepackage{capt-of}

\theoremstyle{plain}
\newtheorem{proposition}{Proposition}

\newcommand\independent{\protect\mathpalette{\protect\independenT}{\perp}}
\def\independenT#1#2{\mathrel{\rlap{$#1#2$}\mkern2mu{#1#2}}}

\hypersetup{
  unicode=false,          
  pdftoolbar=true,        
  pdfmenubar=true,        
  pdffitwindow=false,     
  pdfstartview={FitH},    
  pdftitle={},    
  pdfauthor={Alexander Levis},     
  pdfsubject={Subject},   
  pdfcreator={Alexander Levis},   
  pdfproducer={Alexander Levis}, 
  pdfkeywords={}, 
  pdfnewwindow=true,      
  colorlinks=true,       
  linkcolor=red,          
  citecolor=blue,        
  filecolor=black,      
  urlcolor=cyan           
}

\def\bSig\mathbf{\Sigma}

\title[Double sampling for missing data]{Double sampling for
  informatively missing data in electronic health record-based
  comparative effectiveness research}

\author{Alexander W. Levis$^{1,*}$\email{alevis@cmu.edu}, Rajarshi Mukherjee$^{2}$, Rui Wang$^{2,3}$, \\
\textbf{Heidi Fischer}$\mathbf{^{4}}$ \textbf{and Sebastien Haneuse}$\mathbf{^{2}}$ \\ \\
$^{1}$Department of Statistics \& Data Science, Carnegie Mellon University, Pittsburgh, PA, USA \\
$^{2}$Department of Biostatistics, Harvard T.H. Chan School of Public Health, Boston, MA, USA \\
$^{3}$Department of Population Medicine, Harvard Pilgrim Health Care Institute \\ and Harvard Medical School, Boston, MA, USA \\
$^{4}$Kaiser Permanente Department of Research and Evaluation, 
Pasadena, CA, USA}

\begin{document}


\date{{\it Received October} 2022.}



\pagerange{1--23}
\volume{00}
\pubyear{2021}
\artmonth{October}


\doi{10.1111/j.0000-0000.0000.00000.x}


\label{firstpage}


\begin{abstract}
Missing data arise in most applied settings and are ubiquitous in electronic health records (EHR). When data are missing not at random (MNAR) with respect to measured covariates, sensitivity analyses are often considered. These \textit{post-hoc} solutions, however, are often unsatisfying in that they are not guaranteed to yield concrete conclusions. Motivated by an EHR-based study of long-term outcomes following bariatric surgery, we consider the use of double sampling as a means to mitigate MNAR outcome data when the statistical goals are estimation and inference regarding causal effects. We describe assumptions that are sufficient for the identification of the joint distribution of confounders, treatment, and outcome under this design. Additionally, we derive efficient and robust estimators of the average causal treatment effect under a nonparametric model and under a model assuming outcomes were, in fact, initially missing at random (MAR). We compare these in simulations to an approach that adaptively estimates based on evidence of violation of the MAR assumption. Finally, we also show that the proposed double sampling design can be extended to handle arbitrary coarsening mechanisms, and derive nonparametric efficient estimators of any smooth full data functional.

\vspace{3mm}
 
\end{abstract}

%

\begin{keywords}
  Causal inference; Double sampling; Missing data; Semiparametric theory; Study design
\end{keywords}


\maketitle


%

\section{Introduction}

Missing data is a well-studied problem, with researchers having a vast array of statistical methods at their disposal including inverse-probability weighting (IPW) \citep{seaman2013}, multiple imputation \citep{rubin2004}, and doubly-robust methods \citep{robins1994, tsiatis2007}. For the majority of these, the missing at random (MAR) assumption \citep{rubin1976} is, in one way or another, invoked. For settings where the MAR assumption is not viewed as plausible, methods exist based on alternative sets of identifying assumptions \citep[e.g., ][]{malinsky2020}, availability of an instrumental variable \citep[e.g., ][]{sun2018b} or a ``shadow variable'' \citep[e.g., ][]{miao2016}, sensitivity analyses \citep[e.g., ][]{robins2000} and the estimation of bounds \citep[e.g., ][]{manski1990}. Interestingly, common to all of these methods is that they approach the task of dealing with missing data as a \textit{post-hoc} challenge, that is with an exclusive focus on methods for the data at-hand. 

An alternative strategy is to engage in additional data collection, referred to in this paper as \textit{double-sampling}, specifically to obtain information that could either inform the plausibility of missingness assumptions or be used in an analysis to mitigate bias, or both. Such a strategy is common when addressing confounding \citep{borgan2018} and measurement error and/or missclassification \citep{carroll2006, amorim2021}, but seems to have been under-explored as a strategy for addressing missing data, \citep{hansen1946, frangakis2001, guan2018, miao2021}. Moreover, a general treatment of double sampling for missing data, or more generally coarsened data \citep{heitjan1991}, in the context of causal inference has not been developed.

One important area of biomedical and public health research where missing data is almost ubiquitous is that of studies making use of electronic health records (EHR). With large sample sizes and rich covariate information over extended periods, EHR data represent a significant and cost-effective opportunity \citep{haneuse2017}. Furthermore, these data present a key alternative when randomized clinical trials are not feasible or could not be conducted ethically. EHR systems, however, are typically designed to support clinical and/or billing activities, and not for any particular research agenda. As such, investigators who wish to use EHR data must deal with potential threats to validity including, as mentioned, missing data. Moreover, whether a particular data element is observed in an EHR is likely dependent on the complex interplay of numerous factors \citep{haneuse2016a}, which may, in turn, cast doubt on the plausibility of the MAR assumption. In such settings, augmentation of the EHR with additional information via double sampling may be especially helpful \citep{haneuse2016}. Indeed, \citet{koffman2021} recently reported on a telephone-based survey used to obtain additional information for use in an investigation of the association between bariatric surgery and five-year weight outcomes using data from an EHR \citep[e.g.,][]{arterburn2020}. Key to the latter was the fact that many subjects who had undergone bariatric surgery disenrolled from their health plan before their five-year post-surgery date. Towards understanding the reasons for disenrollment and to evaluate the MAR assumption, the investigators conducted the telephone-based survey to obtain the otherwise missing weight information and other relevant factors. Although their report focuses on disenrollment in relation to missingness, the authors did stress the potential for using the augmented data to correct an otherwise invalid analysis (i.e. of the association between bariatric surgery and weight at three years), but identified the need for novel statistical methods to be developed.

Motivated by this backdrop, we consider double sampling as a means to deal with potentially informatively missing or MNAR data. Specifically, we present novel identification results for the causal average treatment effects in observational settings with missing outcome data. Based on these we describe a suite of five analysis strategies for the context we consider, each distinguished by the nature of the data that is taken to be available, the assumptions that analysts are required to make and the estimator that is to be employed. For the proposed strategies that, to-date, have not been formally described, we establish asymptotic results and characterize efficiency and robustness properties. Finally, we generalize many of these results to allow for arbitrary coarsening of the desired complete data of interest, where complete data are recovered on a subsample via intensive follow-up.  Note, throughout, when not provided in the text, detailed proofs are presented in Appendices A and B.

\section{A hypothetical EHR-based study} \label{sec:hypothetical}

\subsection{Context, notation and terminology}

To anchor the methods we propose, consider a hypothetical EHR-based study for which the goal is to compare two bariatric surgery procedures (e.g., Roux-en-Y gastric bypass vs sleeve gastrectomy) in relation to three-year weight change outcomes. To that end, we assume that appropriate inclusion/exclusion criteria have been specified and operationalized to identify all patients in the EHR who are `eligible' for the study, resulting in a sample of size $n$ which is taken to be a random sample from the population of interest.

Formally, let $A \in \mathcal{A}$, with $|\mathcal{A}| < \infty$, denote the treatment and $Y \in \mathbb{R}$ the outcome of interest. In the hypothetical study, $A$ represents the type of surgery, and $Y$ the change in BMI at three years post-surgery relative to baseline. Then, let $Y(a)$ denote the potential outcome or counterfactual had we fixed treatment level $A = a$, for $a \in \mathcal{A}$. We take the target parameter of interest to be some contrast among the mean counterfactuals, $\mathbb{E}[Y(a)]$. In the hypothetical bariatric surgery study, for example, a natural contrast would be the average treatment effect (ATE), $\mathbb{E}[Y(1)] - \mathbb{E}[Y(0)]$. Throughout this work, towards estimating $\mathbb{E}[Y(a)]$ using the data from the EHR, we invoke the usual `causal' identifying assumptions of consistency, no unmeasured confounding and positivity \citep{hernan2019}. Regarding the control of confounding bias, we assume that a sufficient set of confounders $\boldsymbol{L} \in \mathbb{R}^d$, available in the EHR, has been identified to render the assumption of no unmeasured confounding plausible.

For the setting just described, we refer to $(\boldsymbol{L}, A, Y)$ as the \textit{complete data} and conceive it as arising from some joint distribution, $P_c$. Given an i.i.d sample of size $n$ from $P_c$ one could estimate $\mathbb{E}[Y(a)]$ by, say, targeting the $g$-formula functional, $\chi_a(P_c)\ =\ \mathbb{E}_{P_c}[\mathbb{E}_{P_c}[Y\mid \boldsymbol{L}, A=a]]$, which identifies $\mathbb{E}[Y(a)]$ under the aforementioned causal assumptions \citep{robins1986}. For example, one could use the plugin estimator to estimate the ATE via $\widehat{\mbox{ATE}}\ =\ \chi_1(\widehat{P_c}) - \chi_0(\widehat{P_c})$, relying on the fit of an outcome regression model for $\mathbb{E}_{P_c}[Y\mid \boldsymbol{L}, A=a]$.

To complete the context we consider, we assume that, while $\boldsymbol{L}$ and $A$ are measured on all patients, the outcome is only partially observed; that is for some patients the value of $Y$ is missing. In the hypothetical example this may arise because a patient disenrolled from the health plan prior to the 3-year post-surgery date or because they did not have an encounter within some (reasonable) window of the date. Formally, let $R \in \{0,1\}$ be an indicator for the observance of $Y$ in the EHR; at the outset, therefore, the information that is readily available consists of $n$ i.i.d replicates of $(\boldsymbol{L}, A, R, RY)$, referred to as the \textit{incomplete data}.


\subsection{Analysis strategy \#1}\label{analysis:1}

Given incomplete data, one way forward is to combine a complete data strategy, which one would use had such data been available, with some approach for `dealing' with the missing data. For example, one could combine the use of the $g$-formula indicated above with either inverse-probability weighting based on a model for $R$ or multiple imputation for the missing values of $Y$. In addition to the usual causal assumptions, the validity of such a procedure will hinge on a MAR assumption, such as:
\begin{assumption}[Missing at random outcomes] \label{ass:MAR}
  $R \independent Y \mid \boldsymbol{L}, A$.
\end{assumption}
Crucially, with this assumption in hand, one can proceed as would be done otherwise on the basis of those individuals with $R=1$, since the distribution of $Y \mid \boldsymbol{L}, A$ is the same as $Y \mid \boldsymbol{L}, A, R=1$---that no information on the distribution of $Y \mid \boldsymbol{L}, A, R=0$ is available can be safely ignored. For instance, one can employ the $g$-formula on the basis of the complete-case outcome model $\mathbb{E}[Y \mid \boldsymbol{L}, A = a, R = 1]$.

\subsection{The potential for MNAR}

Suppose, however, that a discussion among the collaborators at the design stage of the study (i.e. at the time the study is being planned and/or a grant/proposal is being written) raises the possibility that the outcome data are MNAR; that is, that Assumption \ref{ass:MAR} may not hold with respect to the baseline covariates $\boldsymbol{L}$ that will be available. In the hypothetical bariatric surgery study, for example, it may be that patients with worse outcomes (in a manner beyond what can be predicted with $\boldsymbol{L}$ and $A$) interact more often with the health care system, and thus have less missing data and/or are less likely to disenroll from their health plan. It is also possible that subjects with worse outcomes are more likely to drop out, perhaps to receive care outside of their original health plan.

The key challenge that a violation of Assumption \ref{ass:MAR} poses is that the distribution of $Y \mid \boldsymbol{L}, A, R=0$ can no longer be safely ignored, and yet there is no information to learn about it. Since MNAR is not testable, the literature on methods for data that are MNAR has generally focused on frameworks for sensitivity analyses and analyses that directly provide bounds on the effect of interest. An alternative to these \textit{post-hoc} approaches, especially if the potential for MNAR is established early in the research process, is to engage in additional data collection efforts that are specifically and preemptively tailored to being able (at least partially) to move `beyond' MNAR.

\section{Double sampling when MNAR is suspected}\label{sec:double}

Central to the proposed work is that follow-up is performed for a subsample of the patients for whom $R = 0$, and that the corresponding (otherwise missing) value of $Y$ is ascertained. While such additional data collection is employed in a wide range of settings, in this paper we follow \citet{frangakis2001} and use the label \textit{double sampling}. Practically, this data collection could be achieved in a number of ways, depending on the context. In some settings, for example, it may be feasible to conduct telephone-based interviews or surveys \citep{haneuse2016, koffman2021}. In other instances, depending on the nature of the missing data, manual chart reviews or natural language processing may be appropriate \citep[e.g.,][]{weiskopf2019}.

\subsection{Notation and terminology}

Let $S \in \{0,1\}$ be an indicator for whether a given patient is selected into the follow-up subsample and outcome data are succesfully obtained. Note, by design, $S \equiv S(1 - R)$; $S$ can only be 1 if $R = 0$ and is equal to 0 deterministically if $R = 1$. With this notation, we refer to $O = (\boldsymbol{L}, A, R, S, (R + S)Y)$ as the \textit{final observed data} and the corresponding joint distribution by $P_o$. Throughout this section, we assume that we observe a random sample $O_1, \ldots, O_n \overset{\mathrm{iid}}{\sim} P_o$.

To complete terminology, we refer to $(\boldsymbol{L}, A, Y, R, S)$ as the \textit{full data} and denote the corresponding joint distribution by $P_f$. Note, both $P_c$ and and $P_o$ are induced by $P_f$ via appropriate marginalization. As will become clear, the use of distinct labels is to help clarify where and how key identifying conditions are employed.

\subsection{Identification}\label{sec:identification}

As with analysis strategy \#1, we proceed by specifying assumptions that permit the identification of the complete data distribution, $P_c$, despite only having access to the final observed data. Specifically, consider the following assumptions:
\begin{assumption}[No informative second-stage selection] \label{ass:ident}
  $S \independent Y \mid \boldsymbol{L}, A, R = 0$.
\end{assumption}
\begin{assumption}[Positivity of second-stage selection probabilities] \label{ass:pos}
  For some $\epsilon > 0$, it holds that $P_o[S = 1 \mid \boldsymbol{L}, A, R = 0] \geq \epsilon,
    P_o$-almost surely.
\end{assumption}
Based on these, consider the following identification result:
\begin{proposition}
\label{prop:ident}
	Under Assumptions \ref{ass:ident} and \ref{ass:pos}, the full data distribution $P_f$ is identified from the final observed data distribution $P_o$.
\end{proposition}
 \begin{proof}
 	Let $p_o = \partial P_o/\partial \mu$ and $p_f = \partial P_f/\partial \mu$ denote the densities for $P_o$ and $P_f$, respectively, both with respect to some dominating measure $\mu$. We can then factor the full data distribution as:
 \begin{align*}
 	p_f(\boldsymbol{L}, A, R, S, Y)\ & =\ p_f(\boldsymbol{L}, A, R, S, RY, (1-R)Y) \\
 		& =\ p_o(\boldsymbol{L}, A, R, S, RY)\ p_f(Y\mid \boldsymbol{L}, A, R=0, S)^{1-R} \\
 		& =\ p_o(\boldsymbol{L}, A, R, S, RY)\ p_o(Y\mid \boldsymbol{L}, A, R=0, S = 1)^{1-R}.
 \end{align*}
with the last of these steps enabled by Assumption \ref{ass:ident}. Since the last expression depends solely on $p_o$, it follows that $P_f$ is identified by $P_o$. Note, we may safely introduce $S = 1$ in the conditioning event by Assumption \ref{ass:pos}.
\end{proof}

\subsection{The proposed framework in practice}
We make several observations. First, implicit to the notation of the proposed double-sampling strategy is that those individuals with $S = 1$ will have complete data. As alluded to, this requires that a subject is selected to be followed-up, and also that their initially missing outcome data is successfully obtained. With this, Assumption~\ref{ass:ident} can be viewed as an MAR-type assumption, specifically in relation to the mechanism underpinning who is selected to be double-sampled and successfully followed-up, and, thus, an alternative to or replacement for the usual MAR assumption regarding $R$.

Second, a crucial distinction from most prior work arises from the specific framing we have adopted; that is, that the discussions that lead to consideration of double sampling occur at the design stage of the broader study. With that framing, investigators will generally have substantially more control over whose data are obtained at the second stage (i.e., $S$) than they would over who has complete data in the EHR (i.e., $R$). Practically, however, depending on the mode of data ascertainment, it may not be that all those who are selected actually have complete data. If, as in \citet{koffman2021}, the mode is a telephone-based survey, then there is no guarantee that all those who are selected will have complete data since individuals may choose not to engage. With this, the plausibility of Assumption \ref{ass:ident} may be compromised; it may be that engagement (and hence completeness of data) remains dependent on the outcome. Other modes of ascertainment, however, such as manual chart review or the reading of an image, do not require direct engagement, so that it is reasonable to foresee that all those with $S = 1$ will indeed have complete data. In general, we believe it is important to distinguish the plausibility of MAR in the EHR from that at the second stage: while the former is likely implausible due to the complexity of interactions with the health system and what measurements get recorded and when, the plausibility of the latter is comparable with that of any prospective cohort study. If anything, access to a rich set of covariates from the EHR may make Assumption~\ref{ass:ident} more plausible. We return to this and other practical issues in the Discussion.

Third, an immediate consequence of Proposition \ref{prop:ident} is that the complete data distribution, $P_c$, is identified and, hence, any functional depending on $P_c$ is identified. Therefore, given Assumptions \ref{ass:ident} and \ref{ass:pos}, one can use the final observed data to estimate quantities of interest, such as mean counterfactuals. This observation, in turn, is the basis for a series of additional analysis strategies proposed in the remainder of this section.

\subsection{Analysis strategy \#2}\label{analysis:2}

Given an i.i.d sample of $O = (\boldsymbol{L}, A, R, S, (R+S)Y)$, Proposition \ref{prop:ident} implies that, if assumptions \ref{ass:ident} and \ref{ass:pos} hold, we can nonparametrically identify $\mathbb{E}[Y(a)]$ via the $g$-formula representation:
\begin{equation}
\label{eqn:tau:AS2}
	\tau_a(P_o) = \mathbb{E}_{P_o}[\mu_{a,R}(\boldsymbol{L})\gamma_a(\boldsymbol{L}) + \mu_{a,S}\{\boldsymbol{L})(1-\gamma_a(\boldsymbol{L})\}],
\end{equation}
with $\gamma_a(\boldsymbol{L}) = P_o(R = 1\mid \boldsymbol{L}, A=a)$, $\mu_{a,R}(\boldsymbol{L}) = \mathbb{E}_{P_o}[Y\mid \boldsymbol{L}, A=a, R=1]$, and $\mu_{a,S}(\boldsymbol{L}) = \mathbb{E}_{P_o}[Y\mid \boldsymbol{L}, A=a, S=1]$. With this representation, one can construct an estimator of the ATE by targeting $\tau_a(P_o)$. One simple approach would be to estimate each component nuisance function, that is $\gamma_a(\boldsymbol{L})$, $\mu_{a,R}(\boldsymbol{L})$ and $\mu_{a,S}(\boldsymbol{L})$, via parametric modeling and combine via an empirical version of expression (\ref{eqn:tau:AS2}). In the following, we derive and characterize a nonparametric efficient and multiply robust estimator $\widehat{\tau}_a$ of $\tau(P_o)$. As we will see, efficient estimation additionally requires a model for the treatment probability $\pi_a(\boldsymbol{L}) = P_o(A = 1 \mid \boldsymbol{L})$, as well as for the double sampling probabilities $\eta_{a, 0}(\boldsymbol{L}) = P_o(S = 1 \mid \boldsymbol{L}, A = a, R = 0)$. Note, as we will prove, this approach has the advantage that under relatively mild conditions, $\sqrt{n}$-rate convergence can still be attained while using flexible machine learning-based models for each nuisance function. The following result is a straightforward application of Proposition A.2 proved in Appendix A.

\begin{theorem} \label{thm:IF}
	Let $\mu_a(\boldsymbol{L}) = \mu_{a,R}(\boldsymbol{L})\gamma_a(\boldsymbol{L}) + \mu_{a,S}(\boldsymbol{L})(1-\gamma_a(\boldsymbol{L}))$. The nonparametric influence function with respect to the maximal tangent space of $\tau_a(P_o)$
  is
  \begin{align*}
  \dot{\tau}_a(O; P_o)
    &= \mu_a(\boldsymbol{L}) - \tau_a(P_o) + \frac{\mathds{1}(A =
      a)}{\pi_a(\boldsymbol{L})}\bigg\{
      \left(R + \frac{S}{\eta_{a,0}(\boldsymbol{L})}\right)Y -
      \mu_a(\boldsymbol{L})
  \\
    & \quad \quad \quad \quad \quad \quad \quad
      \quad \quad \quad \quad \quad \quad \quad
      + (1 - R)\left(1 - \frac{S}{\eta_{a,0}(\boldsymbol{L})}\right)
      \mu_{a, S}(\boldsymbol{L})\bigg\}.
\end{align*}
\end{theorem}

With this efficient influence function in hand, one can proceed by using the standard one-step estimator \citep{bkrw1993, pfanzagl2012}, specifically:
\[
	\widehat{\tau}_a = \tau_a(\widehat{P}_o) + \frac{1}{n}\sum_{i=1}^n \dot{\tau}_a(O; \widehat{P}_o).
\]
For simplicity we assume in the following results that the nuisance functions in $\widehat{P}_o$ are trained on a separate independent sample. In practice, one can use cross-fitting which involves splitting the data into training and test folds, fitting $\widehat{P}_o$ on the training fold, and computing the one-step estimator in the test fold \citep{pfanzagl2012, schick1986}. Full efficiency can be recovered by swapping the roles of the folds and averaging the resulting estimators \citep{chernozhukov2018}. The following result is obtained directly from Theorem A.1 and Proposition A.4 (see Appendix A):
\begin{theorem} \label{thm:asymp}
Suppose
  $\left\lVert \dot{\tau}_a(\, \cdot \, ; \widehat{P}_o) - \dot{\tau}_a(\,
    \cdot \, ; P_o)\right\rVert = o_P(1)$. Then
  \begin{align*}
    \widehat{\tau}_a - \tau_a(P_o) =
    O_P\left(\frac{1}{\sqrt{n}} +
     \mathrm{Bias}_{\tau_a}(\widehat{P}_o; P_o)\right),
  \end{align*}
  where
  \begin{align*}
  \mathrm{Bias}_{\tau_a}(\widetilde{P}_o; P_o)
  &= \mathbb{E}_{P_o}\left(\left(1-
    \frac{\pi_a(\boldsymbol{L})}
    {\widetilde{\pi}_a(\boldsymbol{L})}\right)\left(\widetilde{\mu}_a(\boldsymbol{L})
    - \mu_a(\boldsymbol{L})\right)\right) \\
  & \quad \quad + \mathbb{E}_{P_o}\left((1 - \gamma_a(\boldsymbol{L}))
    \frac{\pi_a(\boldsymbol{L})}
    {\widetilde{\pi}_a(\boldsymbol{L})}\left(1 -
    \frac{\eta_{a, 0}(\boldsymbol{L})}
    {\widetilde{\eta}_{a,0}(\boldsymbol{L})}\right)
    \left(\widetilde{\mu}_{a,S}(\boldsymbol{L})
    - \mu_{a,S}(\boldsymbol{L})\right)\right).
\end{align*}
  for any (fixed) $\widetilde{P}_o$. Moreover,
  if
  $\mathrm{Bias}_{\tau_a}(\widehat{P}_o; P_o) = o_P(n^{-1/2})$, then
  $\sqrt{n}(\widehat{\tau}_a - \tau_a(P_o)) \overset{d}{\to} \mathcal{N}(0,
    V)$, where $V = \mathrm{Var}_{P_o}(\dot{\tau}_a(O; P_o))$ is the
  nonparametric efficiency bound.
\end{theorem}

In addition to consistency, asymptotic normality, and efficiency, Theorem \ref{thm:asymp} reveals a set of robustness properties of $\widehat{\tau}_a$. Specifically, observe that $\mathrm{Bias}_{\tau_a}(\widetilde{P}_o; P_o) = 0$ if: (i) $(\widetilde{\mu}_{a, R}, \widetilde{\mu}_{a, S}, \widetilde{\gamma}_a) = (\mu_{a, R}, \mu_{a, S}, \gamma_a)$; (ii) $(\widetilde{\mu}_{a, S}, \widetilde{\pi}_a) = (\mu_{a, S}, \pi_a)$; or (iii) $(\widetilde{\pi}_a,\widetilde{\eta}_{a, 0}) = (\pi_a, \eta_{a, 0})$. In particular, if the double sampling probabilities $\eta_{a, 0}$ are known by design, then  $\mathrm{Bias}_{\tau_a}(\widetilde{P}_o; P_o) = 0$ if $(\widetilde{\mu}_{a, R}, \widetilde{\mu}_{a, S}, \widetilde{\gamma}_a) = (\mu_{a, R}, \mu_{a, S}, \gamma_a)$ or $\widetilde{\pi}_a = \pi_a$. This robustness extends to the rate of convergence of the estimator (as in \citet{rotnitzky2020}), in that the asymptotic bias term is bounded under mild conditions by a sum of product of $L_2(P_o)$-errors in nuisance function estimation: $\lVert \widehat{\pi}_a - \pi_a \rVert \cdot \lVert \widehat{\mu}_a - \mu_a \rVert$ + $\lVert \widehat{\eta}_{a,0} - \eta_{a, 0} \rVert \cdot \lVert \widehat{\mu}_{a, S} - \mu_{a, S} \rVert$, where $\lVert f \rVert^2 = \mathbb{E}_{P_o}(f(O)^2)$ for any function $f$. In particular, if $\widehat{\pi}_a$, $\widehat{\mu}_{a}$, $\widehat{\eta}_{a,0}$, and $\widehat{\mu}_{a, S}$ are each $L_2(P_o)$-consistent at rate at least $n^{-1/4}$, then (under mild conditions) $\widehat{\tau}_a$ is $\sqrt{n}$-consistent and asymptotically efficient.

Another appealing consequence of the asymptotic normality result in Theorem \ref{thm:asymp} is that simple Wald-type asymptotically valid confidence intervals are immediately available. For example, we can estimate $\mathrm{Var}(\widehat{\tau}_a)$ with $\widehat{\mathrm{Var}}(\widehat{\tau}_a) = \frac{1}{n}\widehat{V} = \frac{1}{n^2}\sum_{i = 1}^n (\dot{\tau}_a(O_i; \widehat{P}_o))^2$, with corresponding Wald-type confidence interval given by $\widehat{\tau}_a \pm z_{1 - \alpha / 2} \sqrt{\widehat{\mathrm{Var}}(\widehat{\tau}_a)}$, where $z_{\alpha}$ denotes the $\alpha$-quantile of the standard normal distribution.

\subsection{Analysis strategy \#3}\label{analysis:3}

A key feature of $\widehat{\tau}_a$ is that there was no need to invoke the usual MAR assumption for $R$; indeed, no assumptions regarding $R$ are invoked. Suppose, however, that following data collection via double-sampling, the investigative team decides that MAR Assumption \ref{ass:MAR} may indeed be plausible. It may be, for example, that new information regarding the mechanisms that underpin $R$ becomes available. Alternatively, suppose the MAR assumption was always viewed as being potentially plausible and that the additional data was collected for the purpose of gaining efficiency in estimating the parameter of interest. In either of these settings, combining Assumption \ref{ass:MAR} with $S \equiv S(1-R)$ and Assumption \ref{ass:ident}, gives that $(R, S) \independent Y \mid \boldsymbol{L}, A$. With this, one can nonparametrically identify $\mathbb{E}[Y(a)]$ via the $g$-formula representation:
\[
	\tau_a^*(P_o) = \mathbb{E}_{P_o}[\mu_{a, \mathrm{MAR}}(\boldsymbol{L})],
\]
where $\mu_{a, \mathrm{MAR}}(\boldsymbol{L}) = \mathbb{E}_{P_o}[Y \mid \boldsymbol{L}, A, R + S = 1]$. As in Section \ref{analysis:3}, while we could construct an estimator of $\tau_a^*(P_o)$ based on a parametric model for the nuisance function $\mu_{a, \mathrm{MAR}}(\boldsymbol{L})$, we derive a semiparametric efficient estimator of $\tau_a^*(P_o)$ that is of a robust augmented IPW form. The following result is proved in Appendix B.

\begin{theorem}\label{thm:IF:3}
  Under Assumptions \ref{ass:MAR} and \ref{ass:ident}, and with $S \equiv S(1-R)$, the semiparametric efficient influence function of $\tau_a^*(P_o)$
  is
  \begin{align*}
  \dot{\tau}_a^*(O; {P}_o)
    &= \mu_{a, \mathrm{MAR}}(\boldsymbol{L}) - \tau_a^*(P_o) + 
\frac{\mathds{1}(A = a)(R + S)}{\pi_a(\boldsymbol{L}) 
\{\gamma_a(\boldsymbol{L})
+ (1 - \gamma_a(\boldsymbol{L}))
\eta_{a, 0}(\boldsymbol{L})\}}
(Y - \mu_{a, \mathrm{MAR}}(\boldsymbol{L})).
\end{align*}
\end{theorem}
With this result, one can again proceed using the standard one-step estimator, specifically:
\[
	\widehat{\tau}_a^* = \tau_a^*(\widehat{P}_o) + \frac{1}{n}\sum_{i=1}^n \dot{\tau}_a^*(O; \widehat{P}_o),
\]
the asymptotic properties of which are established in the following theorem, itself a straightforward consequence of Theorem A.1 and Proposition A.4 in Appendix A.
\begin{theorem}\label{thm:asymp:3}
  Suppose $\left\lVert \dot{\tau}_a^*(\, \cdot \, ; \widehat{P}_o) - \dot{\tau}_a^*(\,
    \cdot \, ; P_o)\right\rVert = o_P(1)$.
Then
\[
    \widehat{\tau}_a^* - \tau_a^*(P_o) =
    O_P\left(\frac{1}{\sqrt{n}} +
     \mathrm{Bias}_{\tau_a^*}(\widehat{P}_o; P_o)\right),\]
  where
 \[
  \mathrm{Bias}_{\tau_a^*}(\widetilde{P}_o; P_o) = \mathbb{E}_{P_o}\left[\left(1-
    \frac{\pi_a(\boldsymbol{L})\{\gamma_a(\boldsymbol{L}) +
(1-\gamma_a(\boldsymbol{L}))\eta_{a, 0}(\boldsymbol{L})\}}
    {\widetilde{\pi}_a(\boldsymbol{L})\{\widetilde{\gamma}_a(\boldsymbol{L}) +
(1-\widetilde{\gamma}_a(\boldsymbol{L}))\widetilde{\eta}_{a, 0}(\boldsymbol{L})\}}\right)\left(\widetilde{\mu}_{a, \mathrm{MAR}}(\boldsymbol{L})
    - \mu_{a, \mathrm{MAR}}(\boldsymbol{L})\right)\right]
\]
  for any $\widetilde{P}_o$. Furthermore, if
  $\mathrm{Bias}_{\tau_a^*}(\widehat{P}_o; P_o) = o_P(n^{-1/2})$, then
  $\sqrt{n}(\widehat{\tau}_a^* - \tau_a^*(P_o)) \overset{d}{\to} \mathcal{N}(0,
    V^*)$, where $V^* = \mathrm{Var}_{P_o}(\dot{\tau}_a^*(O; P_o))$ is the
  semiparametric efficiency bound under the assumptions of Theorem \ref{thm:IF:3}.
  \end{theorem}
Note that analogous comments to those following Theorem \ref{thm:asymp} for $\widehat{\tau}_a$ can be made here for $\widehat{\tau}_a^*$ as well. First, robustness of $\widehat{\tau}_a^*$ follows from the product bias elucidated in Theorem \ref{thm:asymp:3}: $\mathrm{Bias}_{\tau_a^*}(\widetilde{P}_o; P_o) = 0$ if (i) $\widetilde{\mu}_{a, \mathrm{MAR}} = \mu_{a, \mathrm{MAR}}$; or (ii) $(\widetilde{\pi}_a, \widetilde{\gamma}_a, \widetilde{\eta}_{a,0}) = (\pi_a, \gamma_a, \eta_{a,0})$. Moreover, $\widehat{\tau}_a^*$ is $\sqrt{n}$-consistent and asymptotically efficient if each nuisance function estimate is $L_2(P_o)$-consistent at rate at least $n^{-1/4}$. Second, we can estimate $\mathrm{Var}(\widehat{\tau}_a^*)$ with 
$\widehat{\mathrm{Var}}(\widehat{\tau}_a^*) = \frac{1}{n}\widehat{V}^* = \frac{1}{n^2}\sum_{i = 1}^n (\dot{\tau}_a^*(O_i; \widehat{P}_o))^2$, and construct the corresponding Wald-type confidence interval. 

\subsection{Analysis strategy \#4}

The key distinction between analysis strategies $\#2$ and $\#3$ is in relation to whether the MAR Assumption \ref{ass:MAR} holds. In practice there may not be a consensus as to whether it plausibly holds. For example, one collaborator may believe firmly that the outcomes were initially MAR given $(\boldsymbol{L}, A)$ while another may believe that there remains residual dependence of missingness status on the outcome (either directly or through some other, as-yet unmeasured, factor). In this setting, one option may be to conduct a hypothesis test using the observed data to assess Assumption \ref{ass:MAR}; while this assumption is untestable if all one has access to is the initially observed data, under Assumptions \ref{ass:ident} and \ref{ass:pos} the complete data distribution is identified and, in principle, MAR can be tested. Depending on the results of this test, one could report an analysis based on $\widehat{\tau}_a$ which does not rely on Assumption \ref{ass:MAR}, or $\widehat{\tau}_a^*$ which does. While appealing in its simplicity, na\"{i}ve use of the corresponding standard error estimator for the chosen estimator would not account for the uncertainty in the estimator selection. As such, inference will, in general, not be valid. Because of this, we do not derive any theory for this approach but do consider it as a comparator in the simulation study of Section \ref{sec:sim1}.

\subsection{Analysis strategy \#5}

Finally, building on analysis strategy \#4, one could proceed with a more explicitly data-adaptive approach that selects between the two candidate estimators and provides valid post-selection confidence intervals. To this end, we employ the recently developed methods of \citet{rothenhausler2020}. Briefly, as an overview of the method in the one-parameter case, consider estimation of a generic parameter $\theta_0(P)$ when there are $k+1$ asymptotically linear estimators $\widehat{\theta}_0, \widehat{\theta}_1, \ldots, \widehat{\theta}_k$, such that: (i) the `base' estimator $\widehat{\theta}_0$ is $\sqrt{n}$-consistent for $\theta_0(P)$; and (ii) $\widehat{\theta}_j$ is $\sqrt{n}$-consistent for $\theta_j(P)$, for $j = 1, \ldots, k$, where $\theta_j(P)$ may or may not equal $\theta_0(P)$. With this collection of $k+1$ candidate estimators in hand, \citet{rothenhausler2020} proposed an estimator that selects among them by minimizing an estimate of the mean squared error. Formally, let $\widehat{\sigma}_j^2$ be an estimator of the asymptotic variance of $\sqrt{n}(\widehat{\theta}_j - \theta_j(P))$, and $\widehat{\rho}_j^2$ an estimator of the asymptotic variance of $\sqrt{n}(\widehat{\theta}_j - \theta_j(P) - \widehat{\theta}_0 + \theta_0(P))$. Then the procedure selects the $\overline{j}^{th}$ candidate estimator, where $\overline{j} = \underset{j}{\mathrm{arg\,min}}\,\overline{R}(j)$ with $\overline{R}(j) = \max{\{0, (\widehat{\theta}_j - \widehat{\theta}_0)^2 - \widehat{\rho}_j^2 / n\}} + \widehat{\sigma}_{j}^2 / n$. Importantly, \citet{rothenhausler2020} derive asymptotically valid confidence intervals that take into account the uncertainty due to the selection procedure (see their Theorem 4).

Toward applying the \citet{rothenhausler2020} approach to the context of this paper, we take $\widehat{\tau}_a$ of analysis strategy \#2 to be the `base' estimator (since it does not rely on Assumption \ref{ass:MAR}), with $\widehat{\tau}_a^*$ as an alternative estimator that is, in principle, more efficient than $\widehat{\tau}_a$ if Assumption \ref{ass:MAR} does hold. Then, let $\widehat{V}$ = $n \widehat{\mbox{Var}}[\widehat{\tau}_a]$ and  $\widehat{V}^*$ = $n\widehat{\mbox{Var}}[\widehat{\tau}^*_a]$, and define $\widehat{Q} = \frac{1}{n}\sum_{i = 1}^n \{\dot{\tau}_a(O_i; \widehat{P}_o) - \dot{\tau}_a^*(O_i; \widehat{P}_o)\}^2$
as an estimator of $\mbox{Var}_{P_o}[\dot{\tau}_a(O; P_o) - \dot{\tau}^*_a(O; P_o)]$. The final analysis strategy considers the following `data-adaptive' estimator of the mean counterfactual:
\[
	\widehat{\tau}_a^\dag\ := \left\{
	\begin{array}{cl}
		\widehat{\tau}_a & \mbox{if}\ \widehat{V}\ <\ \mbox{max}\left\{n(\widehat{\tau}_a - \widehat{\tau}^*_a)^2 - \widehat{Q},\ 0\right\}\ +\ \widehat{V}^* \\
		\widehat{\tau}^*_a & \mbox{otherwise} \\
	\end{array} \right..
\]
Intuitively, in the present context, one can interpret `large' values of $(\widehat{\tau}_a - \widehat{\tau}_a^*)^2$ as indicating evidence against MAR Assumption \ref{ass:MAR} holding, so that the procedure selects $\widehat{\tau}_a$ as the estimator. Otherwise, the procedure selects $\widehat{\tau}_a^*$. 



\section{Simulations}\label{sec:simulation}
Table \ref{tab:summ} provides a summary of analysis strategies \#1--5 described in Sections \ref{analysis:1} and \ref{sec:double}, delineating them by the nature of the data that is taken to be available, and the assumptions that must hold for the corresponding estimator to be consistent. In this section, we present two simulation studies, conducted to investigate properties of the five strategies.
In the first, we demonstrate the validity of 
the double sampling approach for handling MNAR data,
verify the the robustness properties of the proposed
nonparametric influence function-based estimator
$\widehat{\tau}_a$, and
compare, under differing degrees of violation of MAR, the
bias and variance of $\widehat{\tau}_a$ to $\widehat{\tau}_a^*$ (strategies \#2 and \#3, respectively) as well as approaches that only use the initially observed incomplete data (strategy \#1).
In the second simulation study, we compare in the absence
of model misspecification the performance of strategies \#2--5, over a range of possible violations of MAR.
We also assess the coverage and length
of the proposed confidence intervals for these
estimators. 

\subsection{Robustness, bias and variance}
\label{sec:sim1}
The framing of the simulation study is, 
following our motivating study in Section \ref{sec:hypothetical},
a hypothetical study comparing two bariatric surgery
procedures on long-term weight outcomes. Specifically, we
consider a binary point exposure $A$, taking on
a value of 0 for Roux-en-Y gastric bypass (RYGB) and
1 for vertical sleeve gastrectomy
(VSG), and continuous outcome $Y$ of the proportion weight
change at three years post-surgery. For simplicity, we consider only one confounder, that being
gender, denoted $L_g$. The estimand of interest is taken
to be the ATE, $\tau_1(P_o) - \tau_0(P_o)$.

To help ground the simulation in a real-world setting, 
we used information on 5,693 patients who underwent 
either VSG or RYGB at Kaiser Permanente Washington 
between January 1, 2008, and December 31,
2010. For these patients, complete information was 
available on gender, bariatric surgery procedure, 
and weight outcomes, so that missingness in the outcome
could then be induced by a known
mechanism. We then generated 5,000 simulated datasets of size
$n = 5,693$ under each of three settings, where we varied the strength
of the violation of MAR. Specifically, we proceeded by (i) sampling
directly from the empirical distribution of $L_g$; (ii) generating
$A \mid L_g \sim \mathrm{Bernoulli}(p_1 L_g + p_0(1 - L_g))$, where
$p_0 = 0.20$ and $p_1 = 0.34$ were taken to approximately mirror their
empirical values; (iii) generating $R \mid L_g, A \sim 
\mathrm{Bernoulli}(\mathrm{expit}(\delta_0 +
  \delta_L L_g + \delta_A A + \delta_{LA} L_gA))$, where
$(\delta_0, \delta_L, \delta_A, \delta_{LA}) = (-1.39, 0.09, -0.05,
-0.35)$ and $\mathrm{expit}(x) = \exp{(x)} / (1 + \exp{(x)})$,
inducing a marginal missingness probability $P_o[R = 0] \approx 0.8$;
(iv) generating $Y \mid L_g, A, R \sim \mathcal{N}(\beta_0 + \beta_LL_g + \beta_A A +
  \beta_{RA} R A, \sigma_Y^2)$, where
$(\beta_0, \beta_L, \beta_A) = (-0.24, 0.023, 0.064)$,
$\beta_{RA} \in \{0, 0.016, 0.032\}$ and $\sigma_Y = 0.11$,
approximately mirroring the marginal empirical distribution of $Y$;
and (v) generating
$S \mid L_g, A, R \sim \mathrm{Bernoulli}((1-R)
  \mathrm{expit}{\left\{\zeta_0 + \zeta_LL_g + \zeta_A A + \zeta_{LA}
      L_g A\right\}})$, where
$(\zeta_0, \zeta_L, \zeta_A, \zeta_{LA}) = (-2.2, 0.4, 0.3, 0.25)$,
inducing a marginal double sampling probability
$P_o[S = 1 \mid R = 0] \approx 0.11$. The parameter $\beta_{RA}$
controls the degree to which the MAR assumption is violated: when
$\beta_{RA} = 0.032$, we say there is a ``large'' violation; when
$\beta_{RA} = 0.016$, there is a more ``moderate'' violation; and, when
$\beta_{RA} = 0$, then there is no violation of MAR (i.e,. Assumption \ref{ass:MAR} holds). The labels of ``moderate'' and ``large'' are admittedly somewhat subjective, but we use them as they qualitatively describe the distance between $\tau_1(P_o) - \tau_0(P_o)$, which does not assume MAR, and $\tau_1^*(P_o) - \tau_0^*(P_o)$, which does.

In all scenarios for $\beta_{RA}$, we computed the nonparametric influence
function-based estimator $\widehat{\tau}_1 - \widehat{\tau}_0$, where we plugged in the
maximum likelihood estimators of the true generating models
$\pi_a, \gamma_a, \mu_{a, R}, \mu_{a, S}$ described above, and assumed
the double sampling probabilities $\eta_{a,0}$ were known. To verify
the theoretical robustness of our influence function-based estimator,
we considered misspecifying (i) models
$(\widehat{\mu}_{a,S}, \widehat{\mu}_{a,R}, \widehat{\gamma}_a)$, (ii)
the model $\widehat{\pi}_a$, and (iii) both
$(\widehat{\mu}_{a,S}, \widehat{\mu}_{a,R}, \widehat{\gamma}_a)$ and
$\widehat{\pi}_a$. In particular,
$\widehat{\mu}_{a,S}, \widehat{\mu}_{a,R}$ were 
misspecified by omitting the main effect of $L_g$,
$\widehat{\gamma}_a$ by omitting
the main effect of $A$ and its interaction with $L_g$, and
$\widehat{\pi}_a$, quite drastically, by estimating
$P_o[A = a \mid L_g]$ using $\widehat{P}_o[A = 1 - a \mid L_g]$.

For comparison, we also computed: (1) the estimator
$\widehat{\tau}_1^* - \widehat{\tau}_0^*$, based on the influence
functions $\dot{\tau}_{a}^*$ that are efficient under MAR (analysis strategy \#3); and (2) estimators
that did not make use of the second-stage outcomes (analysis strategy \#1). We
acknowledge that there are very many approaches one might consider for
analyzing the data using only the initially observed data, but decided that a
reasonable analyst might assume MAR, and proceed by targeting $\xi_1(P_o) - \xi_0(P_o)$, where $\xi_a(P_o) = \mathbb{E}_{P_o}[\mu_{a, R}(\boldsymbol{L})]$, with an outcome
regression based estimator based on the $g$-formula (i.e., averaging
$\mu_{a, R}$ over the empirical distribution of $L_g$), an
inverse-probability weighted (IPW) estimator with
missingness-treatment weights $\pi_a\gamma_a$, as described in
\citet{ross2022reflection}, or an augmented-IPW estimator combining both
approaches as in \citet{davidian2005} and \citet{williamson2012}. We pitted each of these estimators (using the correct
models for $\mu_{a, \mathrm{MAR}}$, $\mu_{a, R}$, $\pi_a$, and
$\gamma_a$) against our influence function-based estimator in all
three scenarios.

The results of the simulation study are presented in Figure \ref{fig:res1}. The robustness of the influence-function based estimator
$\widehat{\tau}_1 - \widehat{\tau}_0$ is clearly seen, as unbiased
inference was obtained in all scenarios when all models were correctly
specified, or either
$(\widehat{\mu}_{a,S}, \widehat{\mu}_{a,R}, \widehat{\gamma}_a)$ or
$\widehat{\pi}_a$ was misspecified. When both were misspecified, some
bias was observed in all three MAR violation scenarios. The initial-sample-only
MAR-based estimators had slightly lower variance, but were
substantially biased when there was even a moderate violation of
MAR. The MAR-efficient estimator
$\widehat{\tau}_1^* - \widehat{\tau}_0^*$, as expected, had the lowest
variance of all estimators and was unbiased in the MAR scenario. When
there was a moderate or large violation of MAR, $\widehat{\tau}_1^* - \widehat{\tau}_0^*$
was biased, though less so than the initial-sample-only MAR-based estimators.

\subsection{Inference and assessment of adaptive estimator} \label{sec:sim2}

Within the same simulation framework, we also assessed the
performance of the adaptive estimator (analysis strategy \#5), and evaluated proposed confidence intervals of all the
estimators considered. For each value in a grid of MAR violation
parameters $\beta_{RA} \in [0, 0.04]$, we simulated 5,000 datasets
exactly as in the previous section. In each case, we computed both
$\widehat{\tau}_1 - \widehat{\tau}_0$ and
$\widehat{\tau}_1^* - \widehat{\tau}_0^*$, where all underlying
nuisance models were correctly specified. Based on these, 
we then also computed the adaptive estimator
$\widehat{\tau}_1^\dagger - \widehat{\tau}_0^\dagger$.

Lastly, to show that care is required when using the data to
decide between $\widehat{\tau}_1 - \widehat{\tau}_0$ and
$\widehat{\tau}_1^* - \widehat{\tau}_0^*$,
we contrasted $\widehat{\tau}_1^\dagger - \widehat{\tau}_0^\dagger$ to an ad hoc adaptive estimator (analysis strategy \#4).
For this, we first test the hypothesis that
$\tau_a(P_o) = \tau_a^*(P_o)$ by assessing the magnitude of
the difference $\widehat{\tau}_a - \widehat{\tau}_a^*$. Formally,
under appropriate conditions, Theorem A.1
implies that
$\sqrt{n}(\widehat{\tau}_a - \widehat{\tau}_a^*)
  \xrightarrow[\mathrm{MAR}]{d} \mathcal{N}(0, Q)$,
where
$Q = \mathrm{Var}_{P_o}(\dot{\tau}_a(O; P_o)) +
  \mathrm{Var}_{P_o}(\dot{\tau}_a^*(O; P_o)) - 2 \cdot
  \mathrm{Cov}_{P_o}(\dot{\tau}_a(O; P_o), \dot{\tau}_a^*(O; P_o))$.
  An ad hoc adaptive estimator is simply
  to choose $\widehat{\tau}_a$ if we reject a test of MAR 
  based on this result, i.e., if
$\sqrt{n}\left|\widehat{\tau}_a - 
\widehat{\tau}_a^*\right| \big/
\sqrt{\widehat{Q}} > z_{1 - \alpha/2}$, and otherwise choose
$\widehat{\tau}_a^*$ if we fail to reject. We computed this
estimator across all simulation settings.

To evaluate confidence intervals, we again focused on the three
parameter values $\beta_{RA} \in \{0, 0.016, 0.032\}$. In the 5,000
simulated datasets for each value, we constructed confidence intervals
for the four estimators described above. For
$\widehat{\tau}_1 - \widehat{\tau}_0$ and
$\widehat{\tau}_1^* - \widehat{\tau}_0^*$, we used influence
function-based Wald-type confidence intervals. For $\widehat{\tau}_1^\dagger - \widehat{\tau}_0^\dagger$, we constructed confidence
intervals based on Theorem 4 of \citet{rothenhausler2020}.
For the ad hoc adaptive estimator, we used the
Wald-type interval corresponding to the baseline estimator chosen
according to the hypothesis test --- a naive approach which we expect
will lead to undercoverage.

The results on the grid of $\beta_{RA}$ values are shown in Figure
\ref{fig:res2}. When $\beta_{RA} = 0$ (i.e.,
MAR holds), all estimators are unbiased,
$\widehat{\tau}_1^* - \widehat{\tau}_0^*$ is most efficient, and the
two adaptive estimators have variance somewhere between that of
$\widehat{\tau}_1^* - \widehat{\tau}_0^*$ and
$\widehat{\tau}_1 - \widehat{\tau}_0$. As $\beta_{RA}$ increases, the
bias of $\widehat{\tau}_1^* - \widehat{\tau}_0^*$, which wrongly
assumes MAR, increases roughly linearly. The two adaptive estimators
also inherit some bias due to being pulled away by
$\widehat{\tau}_1^* - \widehat{\tau}_0^*$. Interestingly, when
$\beta_{RA}$ becomes really large, indicating quite a substantial
violation of MAR, the bias of the two adaptive estimators returns back
towards zero, as it becomes increasingly rare for either of these to
select the estimator which assumes MAR.

The results on the focused set of values
$\beta_{RA} \in \{0, 0.016, 0.032\}$ are arranged in Table~\ref{tab:cov}. The confidence interval for
$\widehat{\tau}_1 - \widehat{\tau}_0$ has the appropriate coverage in
all scenarios, as does the interval for
$\widehat{\tau}_1^* - \widehat{\tau}_0^*$ when MAR holds. In the two
MNAR settings, however, $\widehat{\tau}_1^* - \widehat{\tau}_0^*$ is
biased and its confidence interval is off target. The confidence
interval for the adaptive estimator
$\widehat{\tau}_1^\dagger - \widehat{\tau}_0^\dagger$ also appears to
be valid, with perhaps a bit of undercoverage in finite samples for
moderately large values of $\beta_{RA}$. Finally, the naive confidence
intervals of the ad hoc adaptive estimator tend to be overly narrow.

\section{Data application}\label{sec:application}
In this section, we present an analysis of the proposed methods to data from an EHR-based study comparing the effect of RYGB ($A = 1$) versus VSG ($A = 0$) bariatric surgery procedures on percent weight change at three years post-surgery ($Y$). Data were obtained from three health care sites within Kaiser Permanente: Northern California, Southern California, and Washington. Namely, in line with \citet{arterburn2020}, we use data on $n = 13,514$ adult patients who underwent RYGB or VSG between January 2005 and September 2015, with complete weight data at baseline (closest measurement pre-surgery, up to 6 months) and follow-up (closest measurement within $\pm$ 90 days). See Table~\ref{tab:3} for a summary of baseline characteristics. These data comprised the ``complete-cases'' from a larger collection of 30,991 patients for which follow-up outcomes were only partially observed.

We artificially imposed missingness in the outcome $Y$ on the complete-case data according to an MAR mechanism, as well as an MNAR mechanism. To construct a realistic MAR mechanism, we modeled the probability of missingness from the original larger collection of 30,991 patients using the following baseline covariates $\boldsymbol{L}$: baseline weight, health care site, year of surgery, age, gender, race/ethnicity, number of days of health care use in 7-12 month period pre-surgery, number of days hospitalized in pre-surgery year, smoking status, Charlson/Elixhauser comorbidity score, insurance type, clinical statuses for hypertension, coronary artery disease, diabetes, dyslipedemia, retinopathy, neuropathy, and mental health disorders, and use of medicines including, insulin, ACE inhibitors, ARB, statins, other lipid lowering medications, and other antihypertensives. We regressed the indicator for observing $Y$ on $A$ and $\boldsymbol{L}$ via a SuperLearner ensemble (with library $\{\texttt{SL.glm, SL.ranger, SL.rpart}\}$) using the corresponding R package \citep{polley2019}. Next, to construct a MNAR mechanism, we augmented the fitted values $\widehat{\gamma}$ to include dependence on the outcome $Y$:

\vspace{-4mm}
{\small
\[\widehat{\gamma}\ \longmapsto\ \mathrm{expit}\left\{\left(\mathrm{logit}(\widehat{\gamma}) + \widetilde{Y}\left[0.7  + 0.8 (1 - A) - 1.2\,\mathds{1}(\text{diabetes}) + 0.6\, \mathds{1}(\text{non-commercial insurance})\right]\right)\right\},\]
}

\noindent where $\widetilde{Y}$ is standardized BMI change at 3 years. Finally, these models were used to impose missing outcomes on the sample of $n = 13,514$ patients by sampling $R$ according to a Bernoulli with probability given by the fitted values from the models. The resulting marginal probabilities of missingness were 26\% and 28\% in the MAR and MNAR settings, respectively.

In each of the missingness settings described above, we considered collecting a random subsample (i.e., those with $S = 1$) of initially missing outcomes of size 500, 1,000, and 1,500. For each of the six resulting datasets, we computed and compared point estimates and 95\% confidence intervals for analysis strategies \#1, \#2, \#3, and \#5.  For analysis strategy \#1, we targeted $\xi_1(P_o) - \xi_0(P_o)$, where $\xi_a(P_o) = \mathbb{E}_{P_o}(\mu_{a, R}(\boldsymbol{L}))$, with an augmented-IPW estimator as in \citet{davidian2005} and \citet{williamson2012}. Analysis strategies \#2, \#3, and \#5, correspond to estimators $\widehat{\tau}_1 - \widehat{\tau}_0$, $\widehat{\tau}_1^* - \widehat{\tau}_0^*$, and $\widehat{\tau}_1^{\dagger} - \widehat{\tau}_0^{\dagger}$, respectively. As a benchmark, we compare to a standard full-data augmented-IPW estimator that uses outcome data from all $n = 13,514$ patients. For all estimators, we used flexible SuperLearner ensembles for each component nuisance function, with a library of \texttt{SL.glm}, \texttt{SL.ranger}, and \texttt{SL.rpart}. 

Results are summarized in Figure~\ref{fig:res3}. When MAR holds, all estimators perform well with respect to the benchmark analysis, with $\widehat{\tau}_1^* - \widehat{\tau}_0^*$ having smallest variance, as anticipated. On the other hand, the estimators that assume MAR appear to be biased under MNAR, while the nonparametric efficient estimator $\widehat{\tau}_1 - \widehat{\tau}_0$ and adaptive estimator $\widehat{\tau}_1^{\dagger} - \widehat{\tau}_0^{\dagger}$ (which do not assume MAR) are robust to this violation of MAR. As expected, as the second stage subsample size increases, precision improves for analysis strategies \#2, \#3, and \#5, which incorporate this data. 

\section{Double sampling for arbitrary coarsening} \label{sec:missing}

In this paper, we have focused on the specific causal problem outlined
in Section \ref{sec:hypothetical}. 
That said, the nonparametric identification and estimation results
are entirely generic, and do not depend on either the data structure 
of the given problem nor the specific mean counterfactual estimand of interest. In Appendix A, we lay out the notation for arbitrary
coarsening of a given full data structure and show that under a generalization of
Assumptions \ref{ass:ident} and \ref{ass:pos}, double sampling
identifies the complete data
distribution; derive a transformation of the full
data nonparametric influence function of an arbitrary smooth 
functional that yields the observed data nonparametric influence function;
construct influence function-based estimators using sample splitting; and
characterize the asymptotic behavior of these estimators, including
multiple robustness properties.

\section{Discussion}

In summary, this paper proposes a general framework for the use of double-sampling as a means to address potentially MNAR missing data, when scientific interest lies in estimating causal ATEs. Key to the framework is a suite of novel analysis strategies that exploit data arising from the double sampling scheme, coupled with identifying assumptions that guarantee the corresponding estimators to be asymptotically normal, efficient, and robust.

Table \ref{tab:summ} emphasizes that each of the proposed estimators require Assumption \ref{ass:ident} to hold. In any given applied setting, this assumption will need to be carefully evaluated. As indicated in Section \ref{sec:identification}, depending on the context, one practical issue is that selection by the double sampling scheme may not necessarily yield complete data. In such settings, investigators will need to work through the same thought experiments that one usually would for the standard MAR assumption (such as Assumption \ref{ass:ident}) to try to understand why some individuals engage and others do not. If it is felt that engagement remains dependent on outcome status (beyond what is explained by what is known about the design and covariates ($A$, $\boldsymbol{L}$)), then sensitivity analyses or alternative identification schemes may be necessary; this is an on-going area of our work.

A second practical issue is that, even if all those selected actually engage, the data that arises from the double-sampling scheme may be subject to error or recall bias \citep{haneuse2016}. In some settings, the potential for recall bias may be mitigated through the design. In \citet{koffman2021}, for example, the outcome of interest was weight change at three years post-surgery, so the investigators timed the invitation to participants to coincide with the five-year anniversary. Additionally, following the same broad philosophy of this paper, one could directly learn about potential recall bias by including some participants for whom $R$ = 1 in the double-sampling scheme. This, in turn, would enable a comparison between information provided by the patient and what is available in the EHR. How best to do this and use the resulting information, though, are open questions.

Notwithstanding these practical issues and the fact that logistical or financial considerations may altogether preclude the use of double sampling in some settings, we believe the proposed framework presents a new option for researchers as they contend with potentially informative missing or coarsened data. Beyond those mentioned above, there are many opportunities for future work in this vein, including how best to use the available information in the EHR when allocating resources for double-sampling, as well as developing estimators for a broader set of analysis goals, such as mediation, and outcome types, such as time-to-event outcomes.

\section*{Acknowledgements}
The authors gratefully acknowledge support from NIH grant R01 DK128150. 

\bibliographystyle{chicago}
\bibliography{bibliography}

\newpage

\begin{table}[ht]
\caption{Summary of Analysis Strategies \#1--5} \label{tab:summ}
\centering
\begin{tabular}{ccccccc}
Strategy && Data && Assumptions && Estimator \\
\cline{1-1}\cline{3-3}\cline{5-5} \cline{7-7}
\#1	&& $(\boldsymbol{L}, A, R, RY)$		&& (\ref{ass:MAR}) && Standard/ad-hoc\\
\#2	&& $(\boldsymbol{L}, A, R, S, (R+S)Y)$	&& (\ref{ass:ident}) \& (\ref{ass:pos}) && IF-based, $\widehat{\tau}_a$ \\
\#3	&& $(\boldsymbol{L}, A, R, S, (R+S)Y)$	&& (\ref{ass:MAR}) \& (\ref{ass:ident}) && IF-based, $\widehat{\tau}^*_a$ \\
\#4	&& $(\boldsymbol{L}, A, R, S, (R+S)Y)$	&& (\ref{ass:MAR}) \& (\ref{ass:ident}) or (\ref{ass:ident}) \& (\ref{ass:pos}) && $\widehat{\tau}_a$ or $\widehat{\tau}^*_a$ \\
\#5	&& $(\boldsymbol{L}, A, R, S, (R+S)Y)$	&& (\ref{ass:MAR}) \& (\ref{ass:ident}) or (\ref{ass:ident}) \& (\ref{ass:pos}) && $\widehat{\tau}_a^\dag$ \\
\end{tabular}
\end{table}

\begin{table}[ht]
  \caption{Coverage and other statistics of influence function-based
    and adaptive estimators} \label{tab:cov}
  \begin{adjustbox}{width=0.9\columnwidth,center}
    \centering
    \begin{threeparttable}
      \begin{tabular}{l|rrrrrr}
        \hline
        & & \% Bias & Rel. variance & Rel. MSE & \% Coverage & Average length \\ 
        \hline
        $\beta_{RA} = 0$ & $\widehat{\tau}_1 - \widehat{\tau}_0$ & -0.02 & 1.00 & 1.00 & 94.5 & 0.036 \\ 
        & $\widehat{\tau}_1^* - \widehat{\tau}_0^*$ & 0.06 & 0.50 & 0.50 & 95.4 & 0.026 \\ 
        & $\widehat{\tau}_1^\dagger - \widehat{\tau}_0^\dagger$ & 0.04 & 0.80 & 0.80 & 95.8 & 0.035 \\ 
        & Ad hoc adaptive & 0.08 & 0.64 & 0.64 & 93.8 & 0.026 \\
        \hline
        $\beta_{RA} = 0.016$ & $\widehat{\tau}_1 - \widehat{\tau}_0$ & -0.01 & 1.00 & 1.00 & 94.5 & 0.036 \\ 
        & $\widehat{\tau}_1^* - \widehat{\tau}_0^*$ & 10.26 & 0.50 & 1.06 & 81.8 & 0.026 \\ 
        & $\widehat{\tau}_1^\dagger - \widehat{\tau}_0^\dagger$ & 2.87 & 1.08 & 1.12 & 94.1 & 0.036 \\ 
        & Ad hoc adaptive & 5.67 & 1.01 & 1.18 & 82.6 & 0.028 \\
        \hline
        $\beta_{RA} = 0.032$ & $\widehat{\tau}_1 - \widehat{\tau}_0$ & -0.03 & 1.00 & 1.00 & 94.5 & 0.036 \\ 
        & $\widehat{\tau}_1^* - \widehat{\tau}_0^*$ & 19.57 & 0.51 & 2.72 & 45.3 & 0.026 \\ 
        & $\widehat{\tau}_1^\dagger - \widehat{\tau}_0^\dagger$ & 1.69 & 1.26 & 1.28 & 91.3 & 0.037 \\ 
        & Ad hoc adaptive & 4.56 & 1.56 & 1.68 & 73.6 & 0.031 \\ 
        \hline
      \end{tabular}
      \begin{tablenotes}
      \item MSE, mean squared error; \% Bias,
        $100 \times \frac{\mathrm{mean} - (\tau_1(P_o) -
          \tau_0(P_o))}{\tau_1(P_o) - \tau_0(P_o)}$; Rel. variance,
        empirical variance of estimator divided by that of
        $\widehat{\tau}_1 - \widehat{\tau}_0$; Rel. MSE,
        empirical MSE of estimator divided by that of
        $\widehat{\tau}_1 - \widehat{\tau}_0$.
      \end{tablenotes}
    \end{threeparttable}
  \end{adjustbox}
\end{table}

\begin{table}[ht]
\caption{Baseline characteristics of patients who underwent bariatric surgery, 2005–2015} \label{tab:3}
\centering
\resizebox{5.0in}{!}{%
\begin{tabular}{l|l}
\hline
  & \textbf{Surgery patients}\\
\hline
Number & 13514\\
\hline
Sleeve Gastrectomy [surgery type] (\%) & 4659 (34.5)\\
\hline
Health care site (\%) & \\
\hline
\ \ \ \ Washington & 606 (4.5)\\
\hline
\ \ \ \ Northern California & 3484 (25.8)\\
\hline
\ \ \ \ Southern California & 9424 (69.7)\\
\hline
Year of surgery (mean (SD)) & 2009.64 (1.95)\\
\hline
Years of age at surgery (mean (SD)) & 46.29 (11.04)\\
\hline
Age categories (\%) & \\
\hline
\ \ \ \ 1: $\mathrm{Age} < 45$ & 5939 (43.9)\\
\hline
\ \ \ \ 2: $45 \leq \mathrm{Age} < 65$ & 7019 (51.9)\\
\hline
\ \ \ \ 3: $\mathrm{Age} \geq 65$ & 556 (4.1)\\
\hline
Male [gender] (\%) & 2233 (16.5)\\
\hline
Race/ethnicity (\%) & \\
\hline
\ \ \ \ Black & 2471 (18.3)\\
\hline
\ \ \ \ Hispanic & 4167 (30.8)\\
\hline
\ \ \ \ Unknown/Other & 469 (3.5)\\
\hline
\ \ \ \ White & 6407 (47.4)\\
\hline
Days of health care use 7-12 months pre-surgery (mean (SD)) & 9.24 (7.36)\\
\hline
Insulin use (\%) & 1651 (12.2)\\
\hline
Charlson/Elixhauser comorbidity score (\%) & \\
\hline
\ \ \ \ -1 & 2667 (19.7)\\
\hline
\ \ \ \ 0 & 5147 (38.1)\\
\hline
\ \ \ \ 1 & 3249 (24.0)\\
\hline
\ \ \ \ 2 & 2451 (18.1)\\
\hline
Hypertension diagnosis (\%) & 8063 (59.7)\\
\hline
ACE inhibitor use (\%) & 3437 (25.4)\\
\hline
ARB use (\%) & 1199 (8.9)\\
\hline
Other antihypertensive medication (\%) & 5801 (42.9)\\
\hline
Insurance type (\%) & \\
\hline
\ \ \ \ Commercial & 12096 (89.5)\\
\hline
\ \ \ \ Medicaid & 488 (3.6)\\
\hline
\ \ \ \ Medicare & 930 (6.9)\\
\hline
Diabetes status (\%) & 4956 (36.7)\\
\hline
Days hospitalized in year pre-surgery (mean (SD)) & 0.34 (1.81)\\
\hline
Dyslipidemia diagnosis (\%) & 6474 (47.9)\\
\hline
Statin use (\%) & 3667 (27.1)\\
\hline
Other lipid-lowering medication (\%) & 422 (3.1)\\
\hline
Smoking status (\%) & \\
\hline
\ \ \ \ Ever, Self-Report & 4684 (34.7)\\
\hline
\ \ \ \ Never, Self-Report & 7468 (55.3)\\
\hline
\ \ \ \ No Self-Report & 1362 (10.1)\\
\hline
Coronary artery disease (\%) & 355 (2.6)\\
\hline
Mental health diagnoses (\%) & \\
\hline
\ \ \ \ Mild-Moderate Anxiety/Depression & 5822 (43.1)\\
\hline
\ \ \ \ None & 6096 (45.1)\\
\hline
\ \ \ \ Other & 1596 (11.8)\\
\hline
Retinopathy (\%) & 615 ( 4.6)\\
\hline
Neuropathy (\%) & 945 ( 7.0)\\
\hline
Baseline BMI (mean (SD)) & 44.55 (7.12)\\
\hline
\end{tabular}%
}
\end{table}

\begin{figure}[p]
  \centering
  \includegraphics[width = 0.70\linewidth]{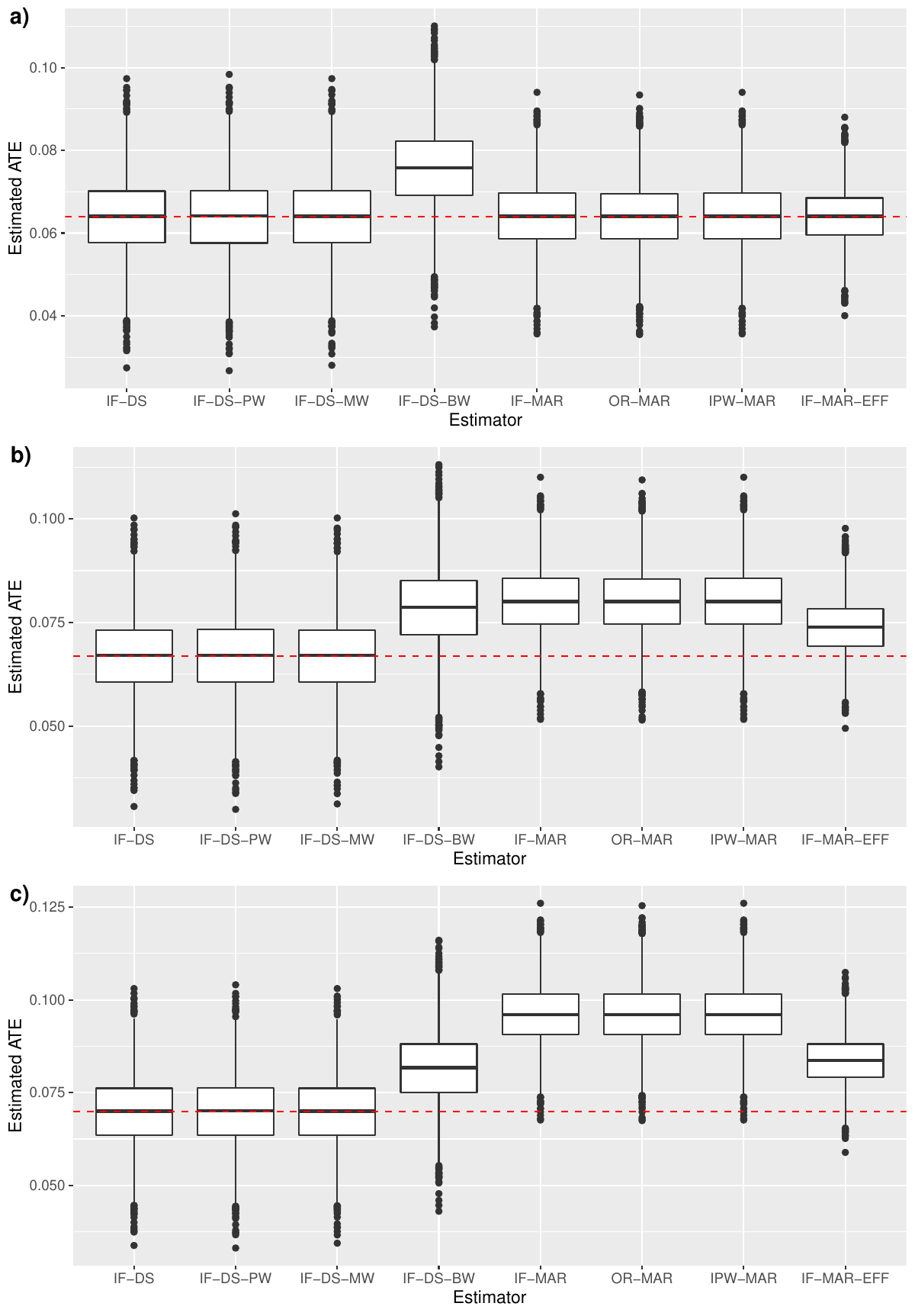}
  \caption{Simulation results for experiments of Section \ref{sec:sim1}, with a) no violation of MAR ($\beta_{RA} = 0$); b) a moderate violation of MAR ($\beta_{RA} = 0.016$); and c) a large violation of MAR
    ($\beta_{RA} = 0.032$). IF-DS = influence-function based estimator
    using double sampling, $\widehat{\tau}_1 - \widehat{\tau}_0$;
    IF-DS-PW = IF-DS but with $\widehat{\pi}_a$ misspecified; IF-DS-MW
    = IF-DS but with
    $(\widehat{\mu}_{a,S}, \widehat{\mu}_{a,R}, \widehat{\gamma}_a)$
    misspecified; IF-DS-BW = IF-DS but with both
    $(\widehat{\mu}_{a,S}, \widehat{\mu}_{a,R}, \widehat{\gamma}_a)$
    and $\widehat{\pi}_a$ misspecified; IF-MAR = augmented IPW-based
    estimator assuming MAR, $\widehat{\xi}_1 - \widehat{\xi}_0$;
    OR-MAR = outcome regression-based estimator assuming MAR; IPW-MAR
    = IPW-based estimator assuming MAR; IF-MAR-EFF = semiparametric
    efficient estimator under MAR,
    $\widehat{\tau}_1^* - \widehat{\tau}_0^*$. ATE = average treatment
    effect; the red dashed line indicates the true
    ATE.} \label{fig:res1}
\end{figure}

\begin{figure}[p]
  \centering
  \includegraphics[width = 0.70\linewidth]{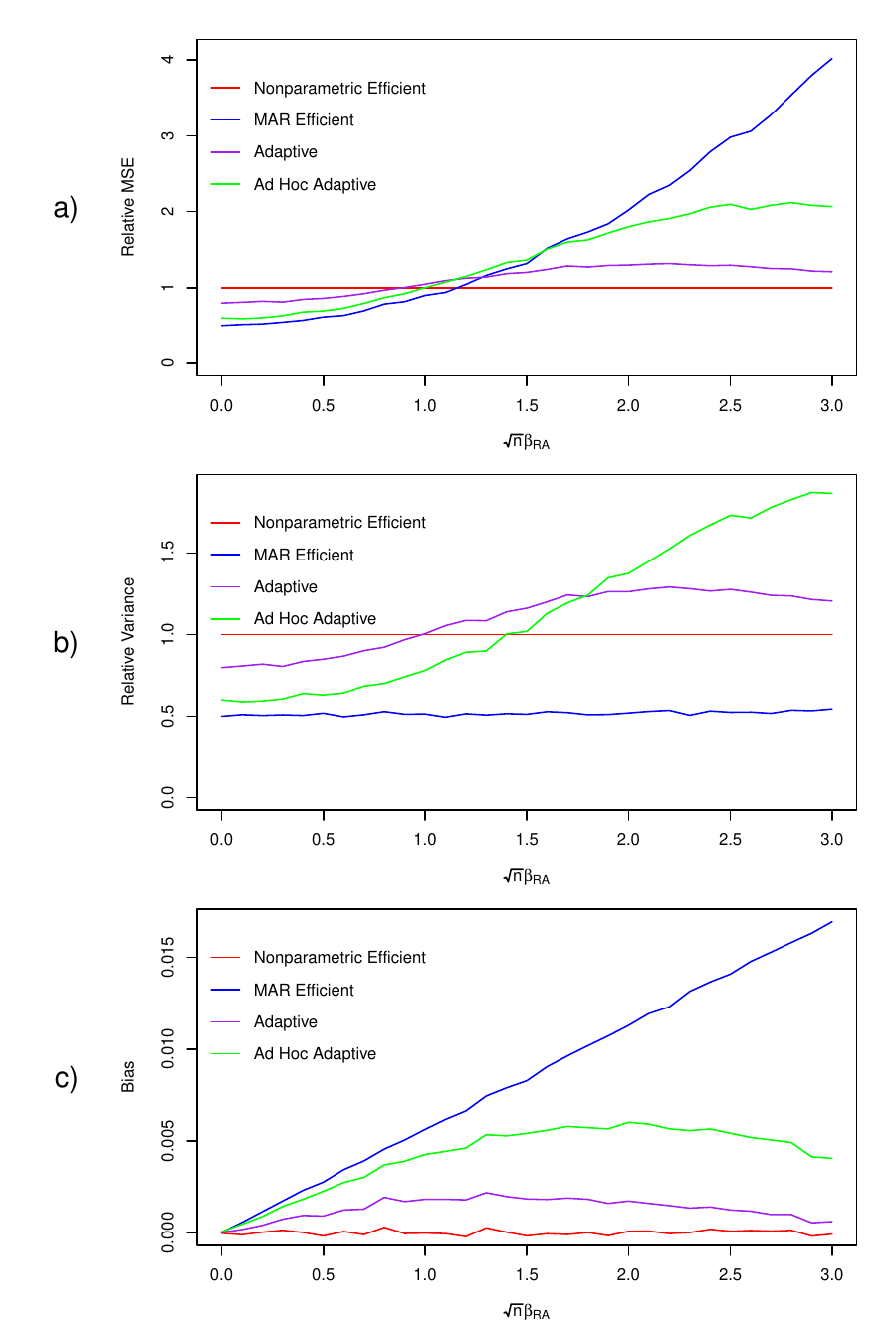}
  \caption{Simulation results for experiments of Section \ref{sec:sim2}, regarding
    $\widehat{\tau}_1 - \widehat{\tau}_0$, MAR semiparametric
    efficient estimator $\widehat{\tau}_1^* - \widehat{\tau}_0^*$,
    the adaptive estimator of \citet{rothenhausler2020}, 
    $\widehat{\tau}_1^{\dagger} - \widehat{\tau}_0^{\dagger}$, and
    the ad hoc approach described in Section \ref{sec:sim1}. Subplot a) shows the
    the empirical mean squared error (MSE) of each approach, divided
    by the MSE of the nonparametric efficient
    estimator. Subplot b) shows the empirical bias of each approach. Subplot c) shows the empirical variance of each approach, divided by the variance
    of the nonparametric efficient estimator.}  \label{fig:res2}
\end{figure}

\begin{figure}[p]
  \centering
  \includegraphics[width = 0.60\linewidth]{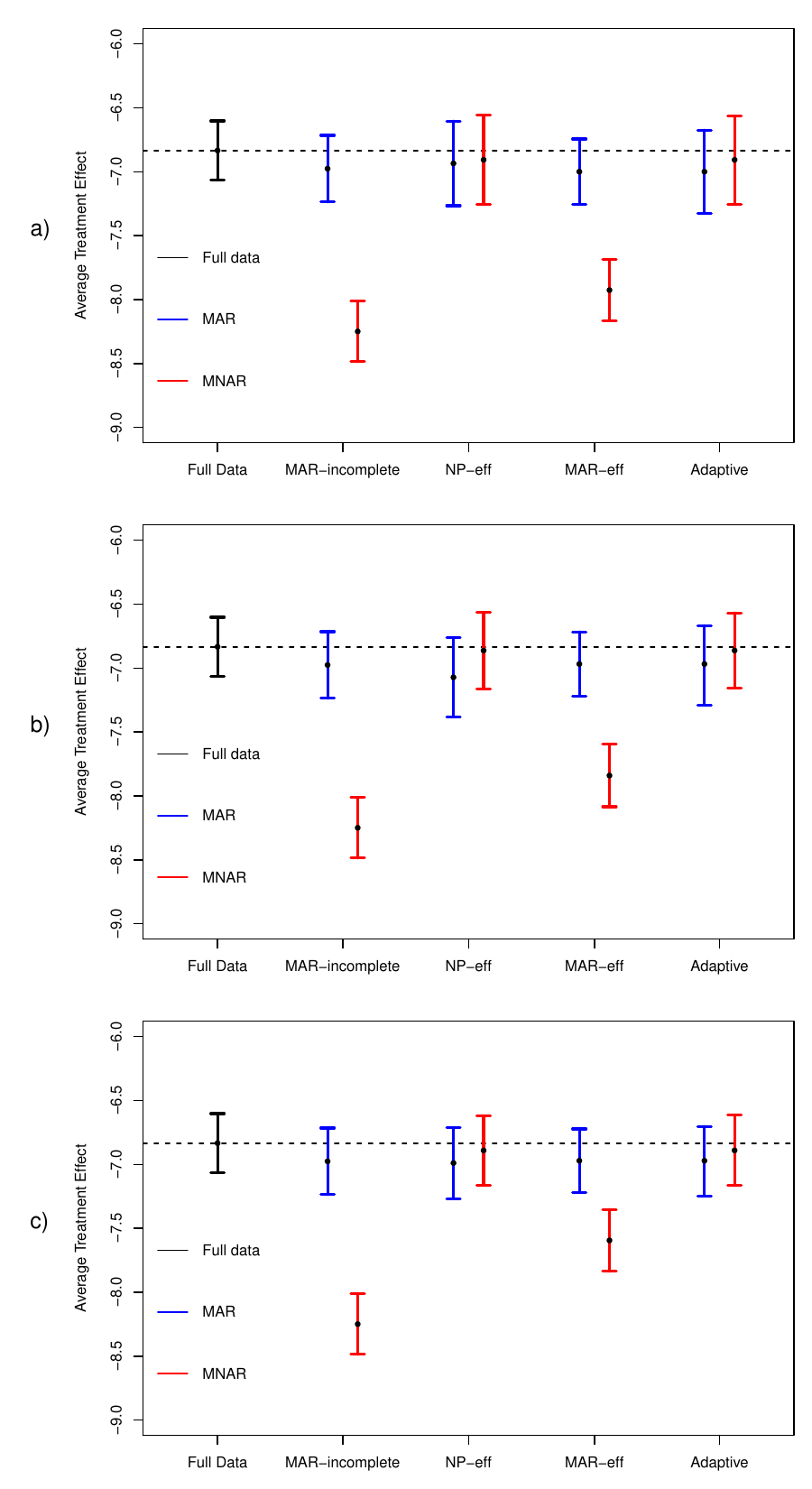}
  \caption{Results for data application in Section~\ref{sec:application}, showing the estimated effect on percent total weight change at 3 years. Subplots a), b), and c) correspond to follow-up subsamples of size 500, 1,000, and 1,500 individuals. Blue and red lines refer to estimated 95\% confidence intervals when outcomes are initially MAR, and MNAR, respectively. Point estimates are marked with black dots. MAR-incomplete = augmented-IPW estimator with incomplete data only, assuming MAR; NP-eff = influence-function based estimator using double sampling, $\widehat{\tau}_1 - \widehat{\tau}_0$; MAR-eff = semiparametric efficient estimator under MAR, $\widehat{\tau}_1^* - \widehat{\tau}_0^*$; Adaptive = adaptive estimator of \citet{rothenhausler2020}, $\widehat{\tau}_1^{\dagger} - \widehat{\tau}_0^{\dagger}$. The black dashed line represents the point estimate for the benchmark complete-case analysis.}  \label{fig:res3}
\end{figure}

\newpage

\begin{appendices}

\section{General Coarsening Framework and Nonparametric Results} \label{app:A}

\subsection{Coarsened data and nonparametric identification}
Suppose the desired \textit{complete} data for a given problem are 
the random vector $\boldsymbol{X} \sim P_{\boldsymbol{X}}^*$. That is, with
$\boldsymbol{X}$ observed on every subject in a random sample,
a parameter of interest, say $\chi(P_{\boldsymbol{X}}^*)$,
could be estimated consistently. Suppose, however, that 
the \textit{initially observed} data consists only of
$(C, \sigma_C(\boldsymbol{X}))$, where $C \in \mathbb{N}$ is a coarsening
random variable, and $\sigma_C(\boldsymbol{X})$ is a coarsened version of the complete data: $\sigma_k$ is some (typically
many-to-one) function for every possible value $k$ of $C$. As in
\citet{tsiatis2007}, we will write $C = \infty$ to denote that the
complete data are observed, i.e., there is no coarsening.  We will further
assume that there exist functions $\overline{\sigma}_k$ for every $k$,
such that $(\sigma_k, \overline{\sigma}_k)$ is injective; that is,
there exist functions $h_k$ with
$h_k(\sigma_k(\boldsymbol{X}), \overline{\sigma}_k(\boldsymbol{X})) =
\boldsymbol{X}$. In our previous observational example from Section
\ref{sec:hypothetical}, 
$C = R \cdot \infty$, where $0 \cdot \infty = 0$, and $\sigma_0(\boldsymbol{X}) = (\boldsymbol{L}, A)$, $\overline{\sigma}_0(\boldsymbol{X}) = Y$.

We now suppose that a subsample is intensively followed up, and
the initially unobserved data 
$\overline{\sigma}_C(\boldsymbol{X})$ are obtained on
some subjects. Let $S \in \{0,1\}$ indicate successful 
follow-up in the subsample when $S = 1$. 
The \textit{full} data are $(C, \sigma_C(\boldsymbol{X}), S,
\overline{\sigma}_C(\boldsymbol{X}))
\sim P^*$, and the \textit{final observed} data are independent and identically
distributed copies of
$O = (C, \sigma_C(\boldsymbol{X}), S, S \cdot
  \overline{\sigma}_C(\boldsymbol{X})) \sim P$. Here, as before, the observed data probability
distribution $P$ and the complete data distribution $P_{\boldsymbol{X}}^*$ are induced by $P^*$. Henceforth,
we suppose that the data at hand are a random sample
$O_1, \ldots, O_n \overset{\mathrm{iid}}{\sim} P$.

Let $p = \frac{d P}{d\mu}$ and $p^* = \frac{d P^*}{d\mu}$ denote the
densities for $P$ and $P^*$, respectively, both with respect to some
dominating measure $\mu$. The density of the full data
distribution can be factored via
    $p^*(C, \sigma_C(\boldsymbol{X}), S,
    \overline{\sigma}_C(\boldsymbol{X})) = p(C, \sigma_C(\boldsymbol{X}), S)\times \prod_{k}
    p^*(\overline{\sigma}_k(\boldsymbol{X}) \mid C = k,
    \sigma_k(\boldsymbol{X}), S)^{\mathds{1}(C = k)}$,
whereas the density of the observed data can be factored via
    $p(O)
    =
    p(C, \sigma_C(\boldsymbol{X}), S)
     \times \prod_{k}
    p(\overline{\sigma}_k(\boldsymbol{X}) \mid
    C = k, \sigma_k(\boldsymbol{X}), S = 1)^{\mathds{1}(S = 1, C
      = k)}$.
As the conditioning event $S = 0$ is possible in the full data but
not in the observed data (i.e., appears in expression
for $p^*$ but not for $p$), $P^*$ will not be
identified from the observed data distribution $P$ unless further
assumptions are made. That said, analysis of the components of the
full data density $p^*$ that are not present in $p$ motivates the
following conditions, which generalize Assumptions \ref{ass:ident} and \ref{ass:pos}.
\begin{assumption}[No informative second-stage selection, general version] \label{ass:identc}
For all $k \neq \infty$, 
  $S \independent \overline{\sigma}_k(\boldsymbol{X}) \mid C = k,
  \sigma_k(\boldsymbol{X})$.
\end{assumption}
\begin{assumption}[Positivity of second-stage sampling probabilities, general version] \label{ass:posc} 
  For some $\epsilon > 0$,
  $P\left[P[S = 1 \mid  C = k,
    \sigma_k(\boldsymbol{X})] \geq \epsilon\right] = 1$,
  for all 
  $k \neq \infty$.
\end{assumption}
By the following result, a generalization of Proposition \ref{prop:ident}, Assumptions \ref{ass:identc} and \ref{ass:posc} are sufficient to identify the full
data distribution $P^*$.

\begin{proposition}\label{prop:identc}
  Assumptions \ref{ass:identc} and \ref{ass:posc} are sufficient to
  identify the full data distribution $P^*$ from the observed data
  distribution $P$.
\end{proposition}
\begin{proof}
  By Assumption \ref{ass:identc}, the full data distribution may be
  factorized via
  \begin{align}
    \begin{split}\label{eq:identc2}
      p^*(C, \sigma_C(\boldsymbol{X}), S,
      \overline{\sigma}_C(\boldsymbol{X})) &= p(C,
      \sigma_C(\boldsymbol{X}), S)
      p^*(\overline{\sigma}_C(\boldsymbol{X})
      \mid C, \sigma_C(\boldsymbol{X}), S) \\
      &= p(C, \sigma_C(\boldsymbol{X}), S)
      p(\overline{\sigma}_C(\boldsymbol{X}) \mid C,
      \sigma_C(\boldsymbol{X}), S = 1)
    \end{split}
  \end{align}
  which only depends on the observed data distribution $p$. Note that
  we may safely introduce $S = 1$ in the conditioning event due to
  Assumption \ref{ass:posc}, and we may ignore the vacuous case $C = \infty$ as
  $\sigma_\infty(\boldsymbol{X}) = \boldsymbol{X}$, given which
  $\overline{\sigma}_{\infty}(\boldsymbol{X})$ is degenerate.
\end{proof}

The identifying Assumption \ref{ass:identc} may be interpreted as
asserting that whether or not a subject is successfully double sampled
is independent of all the initially unobserved data, conditional on
all the initially observed data. While we have circumvented the need
for the usual coarsening at random assumption for the initial sample, it is worth noting that
given only observed data $O$, Assumption \ref{ass:identc} is
untestable. In practice, however, these can be ensured
by certain study designs and the successful follow-up of the chosen subsample: (i) subsample selection completely at random
prior to the study; (ii) subsample selected at random among those with
any initially missing information; and (iii) subsample selected with
investigator-defined probabilities depending only on
$(C, \sigma_C(\boldsymbol{X}))$. Of course, successful follow-up of
the entire intended subsample may not be possible, and in these cases it must be that initially observed data is sufficient to predict successful follow-up. In general, if the same method of contacting subjects is used at the first stage and second stage of data collection, then it may be unreasonable to assume that coarsening at random fails to hold but Assumption \ref{ass:identc} is valid. Thus, the double sampling approach may be most justifiable when the method of data collection at the second stage differs from the first, e.g., in the EHR example, where the initial sample is the data that happened to be recorded in the electronic record, and the subsample is followed up via telephone or in-depth chart review.

On the other hand, Assumption \ref{ass:posc} asserts that there
are no subpopulations, defined by observed data patterns, that are
systematically excluded from the double sampling strategy (other than
those with initially complete data) --- if this were
not the case, one could not learn about the subpopulations with initial missing information
that were not followed up.

\subsection{Estimation of complete data parameters}
Suppose interest lies in estimating the complete data parameter
$\chi(P_{\boldsymbol{X}}^*) \in \mathbb{R}$, viewed as a functional
from a model space of probability distributions on $\boldsymbol{X}$
--- to which $P_{\boldsymbol{X}}^*$ belongs --- to the real line. By
Proposition \ref{prop:identc}, under Assumptions 
\ref{ass:identc} and
\ref{ass:posc}, any complete data functional $\chi(P_{\boldsymbol{X}}^*)$
has a corresponding observed data functional representation
$\tau(P)$. For example, if
$\chi(P_{\boldsymbol{X}}^*) = \mathbb{E}_{P^*}(g(\boldsymbol{X}))$,
for some function $g$, then $
  \chi(P_{\boldsymbol{X}}^*)
  =
    \mathbb{E}_{P}(\mathbb{E}_{P^*}(g(\boldsymbol{X}) \mid C,
    \sigma_C(\boldsymbol{X})))
  =
    \mathbb{E}_P(\mathbb{E}_{P}(g(\boldsymbol{X}) \mid C,
    \sigma_C(\boldsymbol{X}), S = 1)) \eqqcolon \tau(P)$.
It will often be the case that if the complete data $\boldsymbol{X}$ were
completely observed, one would have in mind a valid estimator of
$\chi(P_{\boldsymbol{X}}^*)$. A natural goal is thus to develop a
general procedure that can in a certain sense transform a complete-data
estimator into one that uses only the observed data
$O_1, \ldots, O_n$. The following proposition is the key semiparametric-theoretical result that will facilitate such a
procedure. Note that it can be seen as a special case of the general theory developed in \citet{robins1994}.
\begin{proposition}\label{prop:IFc}
  Suppose $\chi$ is pathwise differentiable \footnote{see e.g. \citet[Chapter 3]{bkrw1993} for precise definitions.}
  with respect to the complete data model at $P_{\boldsymbol{X}}^*$,
  with influence function
  $\dot{\chi}(\boldsymbol{X}; P_{\boldsymbol{X}}^*)$ (with respect to maximal tangent space), and that
  Assumptions \ref{ass:identc} and \ref{ass:posc} hold. Then $\tau(P) = \chi(P_{\boldsymbol{X}}^*)$
  is pathwise differentiable with influence function
  $
    \dot{\tau}(O; P) = \nu_C(\sigma_C(\boldsymbol{X})) +
    \frac{S}{\eta(C,
      \sigma_C(\boldsymbol{X}))}\left\{\dot{\chi}(\boldsymbol{X};
      P_{\boldsymbol{X}}^*) -
      \nu_C(\sigma_C(\boldsymbol{X}))\right\}
  $ at $P$,
  where
  $\nu_C(\sigma_C(\boldsymbol{X})) =
    \mathbb{E}_P(\dot{\chi}(\boldsymbol{X}; P_{\boldsymbol{X}}^*) \mid
    C, \sigma_C(\boldsymbol{X}), S = 1)$, and
  $\eta(C, \sigma_C(\boldsymbol{X})) = P[S = 1\mid C,
  \sigma_C(\boldsymbol{X})]$. Recalling that
  $\sigma_{\infty}(\boldsymbol{X}) = \boldsymbol{X}$, we allow for the
  possibility that $\eta(\infty, \boldsymbol{X}) = 0$, and define
  $\dot{\tau}(O; P)$ to equal
  $\nu_{\infty}(\boldsymbol{X}) = \dot{\chi}(\boldsymbol{X};
  P_{\boldsymbol{X}}^*)$ when $C = \infty$.
\end{proposition}

\begin{proof}
Recall from \citet{bkrw1993} and \citet{tsiatis2007}
that an \textit{influence function}
of a pathwise differentiable functional $\chi(P)$, at $P$
in a given statistical model, is a zero-mean finite-variance
function $\dot{\chi}(O;P)$ of observed data $O$ such that for
any regular parametric submodel $\{P_{\epsilon}: \epsilon \in [0,1)\}$ through $P_{0} \equiv P$, it holds that
\[\left. \frac{d}{d\epsilon}\chi(P_\epsilon)
\right|_{\epsilon = 0} = 
\mathbb{E}_P\left(\dot{\chi}(O;P) g(O)\right),\]
where $g(O)$ is the score function of the parametric 
submodel at $P$. The \textit{tangent set} $\mathcal{T}_P$ of the statistical model 
at $P$ is the set of all score
functions of one-dimensional regular parametric submodels
through $P$, and the \textit{tangent space} is $\Lambda_P = \overline{[\mathcal{T}_P]}$, the
closure of the linear span (with respect to
the Hilbert space $L_2(P)$) of the tangent set. The \textit{efficient influence function} of $\chi$
at $P$ is the unique influence function belonging
to $\Lambda_P$. When $\Lambda_P = L_2(P)$, the model is
said to be nonparametric, and there is a unique influence
function, often called the nonparametric influence function.
See \citet{bkrw1993} and \citet{vandervaart2000} for precise
definitions.

\vspace{2mm}
  Let $\{P_{\epsilon} \mid \epsilon \in [0, 1)\}$ be an arbitrary
  one-parameter regular parametric submodel through $P \equiv
  P_0$. Note that, for any $\epsilon \in [0,1)$,
  \[p_\epsilon(O) = p_\epsilon(C, \sigma_C(\boldsymbol{X}),
    S)p_\epsilon(\overline{\sigma}_C(\boldsymbol{X}) \mid C,
    \sigma_C(\boldsymbol{X}), S = 1)^S,\] so we must have
  \[g(O) = g_{C, \sigma_C(\boldsymbol{X}), S} + S \cdot
    g_{\overline{\sigma}_C(\boldsymbol{X}) \mid C,
      \sigma_C(\boldsymbol{X}), S = 1},\] where $g(O)$ is the score
  function of the submodel, typically
  \[g(O) = \left.\frac{d}{d\epsilon}
      \log{(p_{\epsilon}(O))}\right|_{\epsilon = 0},\] and
  $g_{A \mid B}$ is the conditional score of $A$ given $B$ for
  arbitrary variables $A, B$. Of course, this submodel defines a
  regular parametric submodel
  $\{P_{\epsilon}^* \mid \epsilon \in [0,1)\}$ through the full
  data distribution $P^* \equiv P_{0}^*$ by Assumption \ref{ass:identc}: for
  any $\epsilon \in [0,1)$,
  \[p_{\epsilon}^*(C, \sigma_C(\boldsymbol{X}), S,
    \overline{\sigma}_C(\boldsymbol{X})) = p_\epsilon(C,
    \sigma_C(\boldsymbol{X}),
    S)p_\epsilon(\overline{\sigma}_C(\boldsymbol{X}) \mid C,
    \sigma_C(\boldsymbol{X}), S = 1).\] Next, observe that
  \begin{align*}
    &\mathbb{E}_P\left(\frac{S}{\eta(C, \sigma_C(\boldsymbol{X}))}
      \nu_C(\sigma_C(\boldsymbol{X})) g(O)\right) \\
    &= \mathbb{E}_P\left(\frac{S}{\eta(C, \sigma_C(\boldsymbol{X}))}
      \nu_C(\sigma_C(\boldsymbol{X})) g_{C, \sigma_C(\boldsymbol{X}), S}\right), \\
    &= \mathbb{E}_{P}\left(\frac{S}{\eta(C, \sigma_C(\boldsymbol{X}))}
      \dot{\chi}(\boldsymbol{X}; P_{\boldsymbol{X}}^*)
      g_{C, \sigma_C(\boldsymbol{X}), S}\right).
  \end{align*}
  Here, the first equality results from conditioning on
  $(C, \sigma_C(\boldsymbol{X}), S)$, given which
  $g_{\overline{\sigma}_C(\boldsymbol{X}) \mid C,
    \sigma_C(\boldsymbol{X}), S}$ has mean zero
  ($g_{\overline{\sigma}_C(\boldsymbol{X}) \mid C,
    \sigma_C(\boldsymbol{X}), S = 1}$ can be replaced by
  $g_{\overline{\sigma}_C(\boldsymbol{X}) \mid C,
    \sigma_C(\boldsymbol{X}), S}$ due to the presence of indicator
  $S$). The second equality again can be seen by conditioning
  throughout by $(C, \sigma_C(\boldsymbol{X}), S)$. Thus,
  \begin{align*}
    &\mathbb{E}_P\left(\frac{S}{\eta(C, \sigma_C(\boldsymbol{X}))}
      \left\{\dot{\chi}(\boldsymbol{X}; P_{\boldsymbol{X}}^*) -
      \nu_C(\sigma_C(\boldsymbol{X}))\right\} g(O)\right) \\
    &= \mathbb{E}_{P}\left(\frac{S}{\eta(C, \sigma_C(\boldsymbol{X}))}
      \dot{\chi}(\boldsymbol{X}; P_{\boldsymbol{X}}^*)
      [g(O) - g_{C, \sigma_C(\boldsymbol{X}), S}]\right), \\
    &= \mathbb{E}_{P}\left(\frac{S}{\eta(C, \sigma_C(\boldsymbol{X}))}
      \dot{\chi}(\boldsymbol{X}; P_{\boldsymbol{X}}^*)
      g_{\overline{\sigma}_C(\boldsymbol{X}) \mid C,
      \sigma_C(\boldsymbol{X}), S = 1}\right), \text{ since } S^2 = S, \\
    &= \mathbb{E}_{P^*}\left(\frac{P^*[S = 1 \mid C,
      \sigma_C(\boldsymbol{X}), \overline{\sigma}_C(\boldsymbol{X})]}
      {\eta(C, \sigma_C(\boldsymbol{X}))}
      \dot{\chi}(\boldsymbol{X}; P_{\boldsymbol{X}}^*)
      g_{\overline{\sigma}_C(\boldsymbol{X}) \mid C,
      \sigma_C(\boldsymbol{X}), S = 1}\right), \\
    &= \mathbb{E}_{P^*}\left(
      \dot{\chi}(\boldsymbol{X}; P_{\boldsymbol{X}}^*)
      g_{\overline{\sigma}_C(\boldsymbol{X}) \mid C,
      \sigma_C(\boldsymbol{X}), S = 1}\right).
  \end{align*}
  In the third equality, we introduced $P^*$ as it induces $P$, we
  conditioned on
  $(C, \sigma_C(\boldsymbol{X}),
  \overline{\sigma}_C(\boldsymbol{X}))$, and used the fact that, by
  construction, $\boldsymbol{X}$ is equal to
  $h_C(\sigma_C(\boldsymbol{X}),
  \overline{\sigma}_C(\boldsymbol{X}))$; in the fourth equality, we
  used Assumption \ref{ass:identc}.

  Now, see that
  \begin{align*}
    \mathbb{E}_P\left(
    \nu_C(\sigma_C(\boldsymbol{X})) g(O)\right)
    &= \mathbb{E}_P\left(\nu_C(\sigma_C(\boldsymbol{X}))
      g_{C, \sigma_C(\boldsymbol{X}), S}\right), \\
    &= \mathbb{E}_{P^*}\left(\dot{\chi}(\boldsymbol{X}; P_{\boldsymbol{X}}^*)
      g_{C, \sigma_C(\boldsymbol{X}), S}\right),
  \end{align*}
  where in the first equality we note that
  $S \cdot g_{\overline{\sigma}_C(\boldsymbol{X}) \mid C,
    \sigma_C(\boldsymbol{X}), S = 1}$ has mean zero given
  $(C, \sigma_C(\boldsymbol{X}), S)$, and the second equality can be
  seen by conditioning on $(C, \sigma_C(\boldsymbol{X}), S)$ and again
  using
  $S \independent \overline{\sigma}_C(\boldsymbol{X}) \mid C,
  \sigma_C(\boldsymbol{X})$ under $P^*$.

  Finally, defining
  \[\dot{\tau}(O;P) = \nu_C(\sigma_C(\boldsymbol{X})) +
    \frac{S}{\eta(C, \sigma_C(\boldsymbol{X}))}
    \left\{\dot{\chi}(\boldsymbol{X}; P_{\boldsymbol{X}}^*) -
      \nu_C(\sigma_C(\boldsymbol{X}))\right\},\]
  we have shown that
  \[\mathbb{E}_P\left(
      \dot{\tau}(O;P) g(O)\right) =
    \mathbb{E}_{P^*}\left(\dot{\chi}(\boldsymbol{X};
      P_{\boldsymbol{X}}^*) g^*\right),\] where
  $g^* = g_{C, \sigma_C(\boldsymbol{X}), S} +
  g_{\overline{\sigma}_C(\boldsymbol{X}) \mid C,
    \sigma_C(\boldsymbol{X}), S = 1}$ is the score of the full
  data submodel through $P^*$. But by assumption that $\chi$ is
  regular at $P_{\boldsymbol{X}}^*$,
  \[\mathbb{E}_{P^*}\left(\dot{\chi}(\boldsymbol{X};
      P_{\boldsymbol{X}}^*) g^*\right) = \left. \frac{d}{d\epsilon}
      \chi(P_{\boldsymbol{X},\epsilon}^*) \right|_{\epsilon = 0},\]
  so that
  \[\left. \frac{d}{d\epsilon}
      \tau(P_{\epsilon}) \right|_{\epsilon = 0} = \mathbb{E}_P\left(
      \dot{\tau}(O;P) g(O)\right)\] as
  $\tau(P_{\epsilon}) = \chi(P_{\boldsymbol{X},\epsilon}^*)$ for all
  $\epsilon \in [0,1)$, by construction. By definition, this means
  that $\dot{\tau}(O;P)$ is an influence function for $\tau$ at $P$,
  as claimed.  
\end{proof}

We remark that another way to interpret the term $\frac{S}{\eta(C, \sigma_C(\boldsymbol{X}))}$ in the case that $\eta(\infty, \boldsymbol{X}) = 0$ is
to use the convention
$0 \cdot \infty = 0$, so that the second term drops out.

We are now equipped to define a one-step estimator of
the general functional $\tau(P)$ that
uses its estimated influence function to correct the bias
of a plugin estimator $\tau(\widehat{P})$. Depending on the form of the
parameter and its influence function, certain components of the
observed data distribution may not need to be estimated (e.g., the running causal example of this paper). In general, though, an estimate of
$P_{\boldsymbol{X}}^*$ can be reconstructed by marginalizing an
estimated version of the identified full data density \eqref{eq:identc2} over $(C, S)$. Letting
$\lambda_C(\overline{\sigma}_C(\boldsymbol{X});
\sigma_C(\boldsymbol{X}))$ be the distribution function of
$\overline{\sigma}_C(\boldsymbol{X})$ given
$C, \sigma_C(\boldsymbol{X}), S = 1$, we can use
$\widehat{P}_{\boldsymbol{X}}^*$, $\widehat{\eta}$,
$\widehat{\lambda}$ to obtain
$
  \dot{\tau}(O; \widehat{P}) =
  \widehat{\nu}_C(\sigma_C(\boldsymbol{X})) +
  \frac{S}{\widehat{\eta}(C,
    \sigma_C(\boldsymbol{X}))}\left\{\dot{\chi}(\boldsymbol{X};
    \widehat{P}_{\boldsymbol{X}}^*) -
    \widehat{\nu}_C(\sigma_C(\boldsymbol{X}))\right\}
$,
where
$\widehat{\nu}_C(\sigma_C(\boldsymbol{X})) = \int
  \dot{\chi}(h_C(\sigma_C(\boldsymbol{X}), \boldsymbol{t});
  \widehat{P}_{\boldsymbol{X}}^*) \,
  d\widehat{\lambda}_{C}(\boldsymbol{t} ; \sigma_C(\boldsymbol{X}))$.

We propose to use sample splitting and cross-fitting \citep{chernozhukov2018},
and fit $\widehat{P}_k$ using data $I_k^c$, for $k=1,\ldots,K$.
We then define
$\widehat{\tau}_k = \tau(\widehat{P}_k) + \frac{K}{n} \sum_{i \in I_k}
  \dot{\tau}(O_i; \widehat{P}_k), \text{ for } k = 1, \ldots, K$,
so that the sample-split influence function-based estimator is 
given by 
$\widehat{\tau} = \frac{1}{K}\sum_{k=1}^k \widehat{\tau}_k$.

\subsection{Consistency, asymptotic normality, and robustness} \label{sec:asympc}
The following result is the basis
for consistency, asymptotic normality, nonparametric
efficiency, and multiple robustness
of the proposed nonparametric influence function-based
estimator $\widehat{\tau}$.
\begin{theorem}\label{thm:asympc}
  Suppose
  $\left\lVert \dot{\tau}(\, \cdot \, ; \widehat{P}_k) - \dot{\tau}(\,
    \cdot \, ; P)\right\rVert = o_P(1)$ for $k = 1, \ldots, K$. Then
  \begin{align*}
    \widehat{\tau} - \tau(P) = O_P\left(\frac{1}{\sqrt{n}} +
    \frac{1}{K}\sum_{k=1}^K \mathrm{Bias}_{\tau}(\widehat{P}_k; P)\right),
  \end{align*}
  where
  $\mathrm{Bias}_{\tau}(\widehat{P}; P) = \mathbb{E}_P(\dot{\tau}(O;
  \widehat{P})) + \tau(\widehat{P}) - \tau(P)$ for any
  $\widehat{P}$. Moreover, if
  $\mathrm{Bias}_{\tau}(\widehat{P}_k; P) = o_P(n^{-1/2})$, for
  $k = 1, \ldots, K$, then
  $\sqrt{n}(\widehat{\tau} - \tau(P)) \overset{d}{\to} \mathcal{N}(0,
    V)$, where $V = \mathrm{Var}_P(\dot{\tau}(O; P))$ is the
  nonparametric efficiency bound.
\end{theorem}

\begin{proof}
  For a given subset $k \in \{1, \ldots, K\}$, we can decompose
  the error of $\widehat{\tau}_k$ relative to $\tau(P)$
  via:
  \begin{align*} \widehat{\tau}_k - \tau(P)
  &=
    \frac{K}{n}\sum_{i\in I_k} \dot{\tau}(O_i, P) \\
  &
    \quad \quad + \left\{\mathbb{E}_P(\dot{\tau}(O; \widehat{P}_k)) +
    \tau(\widehat{P}_k) - \tau(P)\right\} \\
  & \quad
    \quad + \left\{\frac{K}{n}\sum_{i\in I_k} \left[\dot{\tau}(O_i;
    \widehat{P}_k) - \dot{\tau}(O_i, P)\right] -
    \mathbb{E}_P(\dot{\tau}(O; \widehat{P}_k))\right\}.
\end{align*}
  Thus, the error of $\widehat{\tau}$ with respect to
  $\tau(P)$ can be decomposed as follows:
  \begin{align*}
    \widehat{\tau} - \tau(P)
    &= \frac{1}{n}\sum_{i=1}^n
      \dot{\tau}(O_i, P) +
      \frac{1}{K}\sum_{k=1}^K
      \mathrm{Bias}_{\tau}(\widehat{P}_k; P) \\
    & \quad \quad +
      \frac{1}{K}\sum_{k=1}^K\left\{\frac{K}{n}\sum_{i
      \in I_k}
      \left[\dot{\tau}(O_i;
      \widehat{P}_k) -
      \dot{\tau}(O_i, P)\right] -
      \mathbb{E}_P(\dot{\tau}(O;
      \widehat{P}_k))\right\}.
  \end{align*}
  By the central limit theorem, the first term is $O_P(n^{-1/2})$ as
  \[\frac{1}{\sqrt{n}}\sum_{i=1}^n \dot{\tau}(O_i, P) \overset{d}{\to}
    \mathcal{N}(0, V),\] 
    where $V = \mathrm{Var}_P(\dot{\tau}(O;P))$. 
    Next, invoking Lemma 2 in
  \citet{kennedy2020b} or \citet{kennedy2020}, for each
  $k = 1, \ldots, K$, as
  $\{O_i : i \in I_k\} \independent \{O_i : i \in I_k^c\}$,
  \begin{align*}
    &\frac{K}{n}\sum_{i
      \in I_k}
      \left[\dot{\tau}(O_i;
      \widehat{P}_k) -
      \dot{\tau}(O_i, P)\right] -
      \mathbb{E}_P(\dot{\tau}(O;
      \widehat{P}_k)) \\
    &= O_P\left(\frac{\sqrt{K}}{\sqrt{n}}
      \left\lVert \dot{\tau}(\, \cdot \, ; \widehat{P}_k) - \dot{\tau}(\,
      \cdot \, ; P)\right\rVert\right), \\
    &= o_P\left(\frac{\sqrt{K}}{\sqrt{n}}\right),
  \end{align*}
  as $n/K \to \infty$, since
  $\left\lVert \dot{\tau}(\, \cdot \, ; \widehat{P}_k) - \dot{\tau}(\,
    \cdot \, ; P)\right\rVert = o_P(1)$ by assumption. Equivalently,
  as $n \to \infty$, this $k$-th term is $o_P(n^{-1/2})$. As $K$ is
  fixed, an average of $K$ such terms is also $o_P(n^{-1/2})$, and the
  result follows.
\end{proof}

It is important to note that we have only established that
$\widehat{\tau}$ is fully (semiparametric) efficient in a
nonparametric model, i.e., when the model tangent space (see
\citet{bkrw1993}, \citet{tsiatis2007}) is $L_2^0(P)$, consisting of
all mean-zero functions of $O$ with finite variance. 
Nonparametric efficiency of
$\widehat{\tau}$ is guaranteed because there is a unique (thus
efficient) influence function in a nonparametric model, its variance
equal to the nonparametric efficiency bound \citep{bkrw1993}. Derivation of the semiparametric efficient observed data
influence function in proper semiparametric models where restrictions are
placed on $P$ will depend on the form of those restrictions
(e.g., when MAR holds as in Analysis \#3), so
we leave characterizations of efficiency over classes
of restrictions on $P$ for future research.

The next two propositions relate the asymptotic variance
and bias term (as defined in Proposition \ref{prop:IFc})
of $\widehat{\tau}$ to that of a complete-data influence function
based estimator.
\begin{proposition}\label{prop:var}
  Let $\dot{\tau}(O; P)$ be the observed data influence function
  defined in Proposition \ref{prop:IFc}, where $P$ is induced by $P^*$
  satisfying Assumptions \ref{ass:identc} and \ref{ass:posc}. Then the observed
  data nonparametric efficiency bound for estimating
  $\tau(P) = \chi(P_{\boldsymbol{X}}^*)$ is given by:
  $\mathrm{Var}_P(\dot{\tau}(O; P))
    = \mathrm{Var}_{P^*}(\dot{\chi}(\boldsymbol{X};
      P_{\boldsymbol{X}}^*)) + \mathbb{E}_P\left(\left(\frac{1}{\eta(C,
      \sigma_C(\boldsymbol{X}))} - 1\right)
      \mathrm{Var}_P(\dot{\chi}(\boldsymbol{X}; P_{\boldsymbol{X}}^*)
      \mid C, \sigma_C(\boldsymbol{X}), S = 1)\right)$.
\end{proposition}

\begin{proof}
  Note that the second summand of the observed data influence function
  $\dot{\tau}(O; P)$ has mean zero given $(C, \sigma_C(X), S)$, as
  \begin{align*}
    &
      S \, \mathbb{E}_P\left(\dot{\chi}(\boldsymbol{X};
      P_{\boldsymbol{X}}^*) -
      \nu_C(\sigma_C(\boldsymbol{X})) \, \middle| \, C, \sigma_C(X), S\right) \\
    &= S\left\{ \mathbb{E}_P\left(
      \dot{\chi}(\boldsymbol{X};
      P_{\boldsymbol{X}}^*) \, \middle| \, C, \sigma_C(\boldsymbol{X}), S = 1\right) -
      \nu_C(\sigma_C(\boldsymbol{X}))\right\} \\
    &= 0,
  \end{align*}
  by definition of $\nu_C(\sigma_C(\boldsymbol{X}))$. It follows that
  the two summands are uncorrelated, and
  \begin{align*}
    \mathrm{Var}_P(\dot{\tau}(O; P))
    &= \mathrm{Var}_P(\nu_C(\sigma_C(\boldsymbol{X}))) \\
    & \quad \quad +
      \mathrm{Var}_P\left(\frac{S}{\eta(C,
      \sigma_C(\boldsymbol{X}))}\left\{\dot{\chi}(\boldsymbol{X};
      P_{\boldsymbol{X}}^*) -
      \nu_C(\sigma_C(\boldsymbol{X}))\right\}\right).
  \end{align*}
  By the law of total variance, and given that the second summand has
  mean zero given $(C, \sigma_C(X), S)$,
  \begin{align*}
    & \mathrm{Var}_P\left(\frac{S}{\eta(C,
      \sigma_C(\boldsymbol{X}))}\left\{\dot{\chi}(\boldsymbol{X};
      P_{\boldsymbol{X}}^*) -
      \nu_C(\sigma_C(\boldsymbol{X}))\right\}\right) \\
    &= \mathbb{E}_P\left(\frac{S}{\eta(C,
      \sigma_C(\boldsymbol{X}))^2}\mathrm{Var}_P\left(\dot{\chi}(\boldsymbol{X};
      P_{\boldsymbol{X}}^*) \mid C, \sigma_C(X), S =1
      \right)\right) \\
    &= \mathbb{E}_P\left(\frac{1}{\eta(C,
      \sigma_C(\boldsymbol{X}))}\mathrm{Var}_P\left(\dot{\chi}(\boldsymbol{X};
      P_{\boldsymbol{X}}^*) \mid C, \sigma_C(X), S =1
      \right)\right),
  \end{align*}
  where in the last equality we conditioned on
  $(C, \sigma_C(\boldsymbol{X}))$. Finally, as
  \[\mathrm{Var}_P\left(\dot{\chi}(\boldsymbol{X};
      P_{\boldsymbol{X}}^*) \mid C, \sigma_C(X), S =1 \right) =
    \mathrm{Var}_{P^*}\left(\dot{\chi}(\boldsymbol{X};
      P_{\boldsymbol{X}}^*) \mid C, \sigma_C(X) \right),\]
  and
  \[\mathbb{E}_P\left(
      \dot{\chi}(\boldsymbol{X}; P_{\boldsymbol{X}}^*) \, \middle| \,
      C, \sigma_C(\boldsymbol{X}), S = 1\right) =
    \mathbb{E}_{P^*}\left( \dot{\chi}(\boldsymbol{X};
      P_{\boldsymbol{X}}^*) \, \middle| \, C,
      \sigma_C(\boldsymbol{X})\right),\]
  by Assumption \ref{ass:identc}, adding and subtracting
  \[\mathbb{E}_P\left(\mathrm{Var}_P\left(\dot{\chi}(\boldsymbol{X};
        P_{\boldsymbol{X}}^*) \mid C, \sigma_C(X), S =1
      \right)\right)\] yields the result, again by the law of total
  variance.
\end{proof}

 \begin{proposition} \label{prop:bias} 
 Let $P$, $\widetilde{P}$ be arbitrary
  putative observed data distributions, induced by full data
  distributions $P^*$, $\widetilde{P}^*$, respectively, and assume
  that both $P^*$ and $\widetilde{P}^*$ satisfy Assumptions
  \ref{ass:identc} and \ref{ass:posc}. Then
      $\mathrm{Bias}_{\tau}(\widetilde{P}; P)   \coloneqq \mathbb{E}_P(\dot{\tau}(O; \widetilde{P})) +
      \tau(\widetilde{P}) - \tau(P)$ is equal
      to
      $\mathrm{Bias}_{\chi}(\widetilde{P}_{\boldsymbol{X}}^*;
      P_{\boldsymbol{X}}^*) + 
      \mathbb{E}_{P}\left[\left(1 - \frac{\eta(C, \sigma_C(\boldsymbol{X}))}
          {\widetilde{\eta}(C, \sigma_C(\boldsymbol{X}))}\right)
        \left\{\widetilde{\nu}_C(\sigma_C(\boldsymbol{X})) -
          \check{\nu}_C(\sigma_C(\boldsymbol{X}))
        \right\}\right]$,
  where
  $\mathrm{Bias}_{\chi}(\widetilde{P}_{\boldsymbol{X}}^*;
  P_{\boldsymbol{X}}^*)$ is similarly defined to equal
  $\mathbb{E}_{P^*}(\dot{\chi}(\boldsymbol{X};
  \widetilde{P}_{\boldsymbol{X}}^*)) +
  \chi(\widetilde{P}_{\boldsymbol{X}}^*) - \chi(P_{\boldsymbol{X}}^*)$,
  $\widetilde{\nu}_C(\sigma_C(\boldsymbol{X})) \coloneqq
    \mathbb{E}_{\widetilde{P}}(\dot{\chi}(\boldsymbol{X};
    \widetilde{P}_{\boldsymbol{X}}^*) \mid C,
    \sigma_C(\boldsymbol{X}), S = 1)$,  and
  $\check{\nu}_C(\sigma_C(\boldsymbol{X})) \coloneqq
    \mathbb{E}_P(\dot{\chi}(\boldsymbol{X};
    \widetilde{P}_{\boldsymbol{X}}^*) \mid C, \sigma_C(\boldsymbol{X}),
    S = 1)$.
\end{proposition}
\begin{proof}
  Observe that
  \[\mathbb{E}_P\left(\check{\nu}_C(\sigma_C(\boldsymbol{X}))\right)
    = \mathbb{E}_{P^*}(\dot{\chi}(\boldsymbol{X};
    \widetilde{P}_{\boldsymbol{X}}^*)),\] by iterated expectations and
  Assumption \ref{ass:identc}. Moreover,
  \begin{align*}
    &\mathbb{E}_P\left(\frac{S}{\widetilde{\eta}(C, \sigma_C(\boldsymbol{X}))}
      \left\{\dot{\chi}(\boldsymbol{X}; \widetilde{P}_{\boldsymbol{X}}^*) -
      \widetilde{\nu}_C(\sigma_C(\boldsymbol{X})) \right\}\right) \\
    &= \mathbb{E}_{P^*}\left(\frac{\eta(C, \sigma_C(\boldsymbol{X}))}
      {\widetilde{\eta}(C, \sigma_C(\boldsymbol{X}))}
      \left\{\dot{\chi}(\boldsymbol{X}; \widetilde{P}_{\boldsymbol{X}}^*) -
      \widetilde{\nu}_C(\sigma_C(\boldsymbol{X})) \right\}\right), \\
    &= \mathbb{E}_{P}\left(\frac{\eta(C, \sigma_C(\boldsymbol{X}))}
      {\widetilde{\eta}(C, \sigma_C(\boldsymbol{X}))}
      \left\{
      \check{\nu}_C(\sigma_C(\boldsymbol{X})) -
      \widetilde{\nu}_C(\sigma_C(\boldsymbol{X}))\right\} \right),
  \end{align*}
  where in the first equality we condition on
  $(C, \sigma_C(\boldsymbol{X}),
  \overline{\sigma}_C(\boldsymbol{X}))$, and in the second equality we
  condition on $(C, \sigma_C(\boldsymbol{X}))$ and use
  Assumption \ref{ass:identc}. Hence,
  \begin{align*}
    &\mathbb{E}_P(\dot{\tau}(O; \widetilde{P})) \\
    &= \mathbb{E}_P\left(\widetilde{\nu}_C(\sigma_C(\boldsymbol{X})) +
      \frac{S}{\widetilde{\eta}(C, \sigma_C(\boldsymbol{X}))}
      \left\{\dot{\chi}(\boldsymbol{X}; \widetilde{P}_{\boldsymbol{X}}^*) -
      \widetilde{\nu}_C(\sigma_C(\boldsymbol{X})) \right\}\right), \\
    &= \mathbb{E}_P\left(\check{\nu}_C(\sigma_C(\boldsymbol{X}))\right) +
      \mathbb{E}_P\left(\widetilde{\nu}_C(\sigma_C(\boldsymbol{X})) -
      \check{\nu}_C(\sigma_C(\boldsymbol{X}))\right)\\
    & \quad \quad +
      \mathbb{E}_{P}\left(\frac{\eta(C, \sigma_C(\boldsymbol{X}))}
      {\widetilde{\eta}(C, \sigma_C(\boldsymbol{X}))}
      \left\{
      \check{\nu}_C(\sigma_C(\boldsymbol{X})) -
      \widetilde{\nu}_C(\sigma_C(\boldsymbol{X}))\right\} \right), \\
    &= \mathbb{E}_{P^*}(\dot{\chi}(\boldsymbol{X};
      \widetilde{P}_{\boldsymbol{X}}^*)) \\
    & \quad \quad +
      \mathbb{E}_{P}\left[\left(1 - \frac{\eta(C, \sigma_C(\boldsymbol{X}))}
      {\widetilde{\eta}(C, \sigma_C(\boldsymbol{X}))}\right)
      \left\{\widetilde{\nu}_C(\sigma_C(\boldsymbol{X})) -
      \check{\nu}_C(\sigma_C(\boldsymbol{X}))
      \right\}\right].
  \end{align*}
  The result is obtained by noticing that
  $\chi(P_{\boldsymbol{X}}^*) = \tau(P)$,
  $\chi(\widetilde{P}_{\boldsymbol{X}}^*) = \tau(\widetilde{P})$.
\end{proof}

\begin{remark}
  By Proposition \ref{prop:bias}, we expect the observed-data influence
  function-based estimator $\widehat{\tau}$ to at least partially
  inherit robustness properties of the complete data influence
  function-based estimator $\widehat{\chi}$. The second term in the
  bias expression can be rewritten
  \begin{equation}
  \mathbb{E}_{P^*}\left[\left(1 -
        \frac{\eta(C, \sigma_C(\boldsymbol{X}))} {\widetilde{\eta}(C,
          \sigma_C(\boldsymbol{X}))}\right)
      \left(\frac{\widetilde{\lambda}_{C}'(\overline{\sigma}_C(\boldsymbol{X});
          \sigma_C(\boldsymbol{X}))}
        {\lambda_{C}'(\overline{\sigma}_C(\boldsymbol{X});
          \sigma_C(\boldsymbol{X}))} - 1 \right)
      \dot{\chi}(\boldsymbol{X};
      \widetilde{P}_{\boldsymbol{X}}^*)\right],
  \end{equation}
  where $\lambda_k'$, $\widetilde{\lambda}_k'$ are conditional
  densities of $\overline{\sigma}_k(\boldsymbol{X})$ given
  $ C = k, \sigma_k(\boldsymbol{X}), S = 1$ under $P$ and
  $\widetilde{P}$, respectively, for any $k \neq \infty$. This term
  can be simplified in certain examples (e.g., in our running example), but generally also exhibits a double robust
  property: it is zero if either $\widetilde{\eta}= \eta$ or
  $\widetilde{\lambda} = \lambda$. Thus, when the double sampling
  probabilities $\eta(C, \sigma_{C}(\boldsymbol{X}))$ are known by
  design, this term is automatically zero, so that
  $\mathrm{Bias}_{\tau}(\widetilde{P}; P) =
  \mathrm{Bias}_{\chi}(\widetilde{P}_{\boldsymbol{X}}^*;
  P_{\boldsymbol{X}}^*)$.
\end{remark}

The variance formula of Proposition \ref{prop:var}
facilitates
an analysis of the loss of efficiency that the estimator based on
double sampling incurs, compared to a complete-data influence function-based 
estimator that uses $\boldsymbol{X}$ on the complete sample. 
On the other hand, in view of Theorem \ref{thm:asympc}, the 
bias
formula in Proposition \ref{prop:bias} 
is essential for determining conditions under which
$\widehat{\tau}$ will be asymptotically normal and with variance
attaining the nonparametric efficiency bound.

To elaborate on the previous point, for many common functionals (e.g.,
the running example in Section \ref{sec:hypothetical}), the complete data
asymptotic bias term
$\mathrm{Bias}_{\chi}(\widehat{P}_{\boldsymbol{X}}^*;
P_{\boldsymbol{X}}^*)$ exhibits a ``mixed bias'' property
\citep{robins2008, rotnitzky2020}, in that it involves the product of
nuisance function estimation errors. Moreover, the additional term in
the bias expression also has this property: it is zero if either
$\widehat{\eta} = \eta$ or $\widehat{\lambda}_k = \lambda_k$ for all
$k \neq \infty$. In particular, this additional term is guaranteed to
be zero when the double sampling probabilities $\eta$ are known by
design. Thus, the proposed observed data influence function-based
estimators inherit any robustness properties of their complete data
counterparts when $\eta$ is known, and otherwise will have a slightly
more elaborate multiple robustness structure due to the additional
term.

Another important consequence of the mixed bias property is that of
``rate double robustness'' \citep{rotnitzky2020}. Specifically, in the
complete data setting, if
$\mathrm{Bias}_{\chi}(\widehat{P}_{\boldsymbol{X}}^*;
P_{\boldsymbol{X}}^*)$ depends on the product of $L_2(P)$-norm errors
for estimating a pair of nuisance functions, and if this product
converges to zero at rate $n^{-1/2}$, then the complete data influence
function-based estimator will achieve $n^{-1/2}$ rate inference. The
benefit is that this allows for more flexible estimation (e.g., errors
converging faster than $n^{-1/4}$) of each of the nuisance functions,
and the required convergence rates can be achieved by state-of-the-art
machine learning models under smoothness or sparsity conditions, for
example. In our case, we require the additional property that either
$\eta$ is known, or else
$\lVert \widehat{\eta} -\eta\rVert \lVert
\widehat{\lambda'}/\lambda' - 1 \rVert = o_P(n^{-1/2})$, 
where $\lambda'$ is the density corresponding to $\lambda$.

Finally, by the asymptotic normality result of Theorem
\ref{thm:asympc}, a simple asymptotic variance estimator can obtained
from the empirical variance of the estimated influence functions. The variance estimate for
general functional $\tau(P)$ is
$
  n \cdot \widehat{\mathrm{Var}}(\widehat{\tau}) = \frac{1}{n}\sum_{k
    = 1}^K \sum_{i \in I_k} (\dot{\tau}(O_i; \widehat{P}_k))^2$.
Further, a Wald-type confidence interval is given by
$\widehat{\tau} \pm z_{1 - \alpha / 2}
  \sqrt{\widehat{\mathrm{Var}}(\widehat{\tau})},$ with $z_{\alpha}$
denoting the $\alpha$-quantile of the standard normal distribution.

\section{Semiparametric results
under Assumption \ref{ass:MAR}} \label{app:B}
The observed data density is given by
\[p_o(O) = p_o(\boldsymbol{L}, A, R) 
p_o(Y \mid \boldsymbol{L}, A, R = 1)^R
p_o(S \mid \boldsymbol{L}, A, R = 0)^{1-R}
p_o(Y \mid \boldsymbol{L}, A, S = 1)^S,\]
by Assumption \ref{ass:ident} and the assertion that 
$R = 1$ implies 
$S = 0$, i.e., $S \equiv S(1 - R)$.
Under the MAR assumption (i.e., Assumption \ref{ass:MAR}, 
$R \independent Y \mid \boldsymbol{L}, A$), and Assumption
\ref{ass:ident}, 
$S \independent Y \mid \boldsymbol{L}, A, R = 0$,
we can conclude that 
\begin{equation}\label{eq:MAR2}
    (R, S) \independent Y \mid \boldsymbol{L}, A.
\end{equation} 
As a result, the conditional
densities $p_o(Y \mid \boldsymbol{L}, A, R = 1)$ and 
$p_o(Y \mid \boldsymbol{L}, A, S = 1)$ are equal 
to $p_c(Y \mid \boldsymbol{L}, A) = 
p_o(Y \mid \boldsymbol{L}, A, R + S = 1)$. Thus, the final observed data
distribution $P_o$ belongs to the semiparametric model induced by
MAR if and only if its density can factorized according to
\[p_o(O) = p_o(\boldsymbol{L}, A, R) 
p_o(S \mid \boldsymbol{L}, A, R = 0)^{1 - R}
p_o(Y \mid \boldsymbol{L}, A, R + S = 1)^{R + S}.\]
By Lemma 24 of \citet{rotnitzky2020b}, the tangent space of
the semiparametric model at $P_o$ is
\[\Lambda_{P_o} = \Lambda_{\boldsymbol{L}} \oplus 
\Lambda_{A \mid \boldsymbol{L}} \oplus 
\Lambda_{R \mid \boldsymbol{L}, A} \oplus
(1-R)\Lambda_{S \mid \boldsymbol{L}, A, R}
\oplus (R + S) \Lambda_{Y \mid \boldsymbol{L}, A, R + S},\]
where for any random vectors $\boldsymbol{W}, \boldsymbol{V}$,
$\Lambda_{\boldsymbol{W} \mid \boldsymbol{V}} = 
\left\{a(\boldsymbol{W}, \boldsymbol{V}) \in L_2(P_o):
\mathbb{E}_{P_o}(a(\boldsymbol{W}, \boldsymbol{V}) \mid 
\boldsymbol{V}) = 0\right\}$.
Now, under the MAR semiparametric model, 
$\tau_a^*(P_o) = 
\mathbb{E}_{P_o}(\mu_{a, \mathrm{MAR}}(\boldsymbol{L})) = \tau_a(P_o)$,
where $\mu_{a, \mathrm{MAR}}(\boldsymbol{L}) = 
\mathbb{E}_{P_o}(Y \mid \boldsymbol{L}, A = a, R + S = 1)$, as
\begin{align*}
    \mu_{a}(\boldsymbol{L}) 
    &= \mu_{a,
  R}(\boldsymbol{L})\gamma_a(\boldsymbol{L}) + \mu_{a,
  S}(\boldsymbol{L}) (1 - \gamma_a(\boldsymbol{L})) \\
  &= \mu_{a, \mathrm{MAR}}(\boldsymbol{L}) 
  \gamma_a(\boldsymbol{L}) + 
  \mu_{a, \mathrm{MAR}}(\boldsymbol{L})  (1 - \gamma_a(\boldsymbol{L})) \\
  &= \mu_{a, \mathrm{MAR}}(\boldsymbol{L}),
\end{align*}
by \eqref{eq:MAR2}. Moreover, the nonparametric influence
function of $\tau_a^*(P_o)$ is simply
\[\dot{\tau}_{a, \mathrm{MAR}}(O;P_o) = 
\mu_{a, \mathrm{MAR}}(\boldsymbol{L}) - \tau_a^*(P_o) +
\frac{T}{P_o[T = 1 \mid \boldsymbol{L}]}
(Y - \mu_{a, \mathrm{MAR}}(\boldsymbol{L})),\]
where $T = (R+S)\mathds{1}(A = a)$ --- viewing $T$
as a modified treatment indicator, the influence function
must be of the same form as that for the usual counterfactual mean functional \citep{hahn1998}. The modified
treatment probability can be expanded to 
\begin{align*}
    P_o[T = 1 \mid \boldsymbol{L}] 
    &= P_o[A = a \mid \boldsymbol{L}]
P_o[R = 1 \vee S = 1 \mid \boldsymbol{L}, A = a] \\
&=
\pi_a(\boldsymbol{L})\{\gamma_a(\boldsymbol{L}) +
(1-\gamma_a(\boldsymbol{L}))\eta_{a, 0}(\boldsymbol{L})\}.
\end{align*}
Finally, we notice that $\dot{\tau}_{a, \mathrm{MAR}}(O;P_o)$ must
be the efficient influence function under $P_o$, as it belongs to $\Lambda_{P_o}$. To see
this, note that 
$\mu_{a, \mathrm{MAR}}(\boldsymbol{L}) - \tau_a^*(P_o)$ is a 
mean-zero function of $\boldsymbol{L}$, so belongs to 
$\Lambda_{\boldsymbol{L}}$, and 
$\frac{T}{P_o[T = 1 \mid \boldsymbol{L}]}
(Y - \mu_{a, \mathrm{MAR}}(\boldsymbol{L}))$ has mean zero
given $\boldsymbol{L}, A, R + S$, so belongs to
$(R + S) \Lambda_{Y \mid \boldsymbol{L}, A, R + S}$.

Suppose now that we are not willing to assume that $P_o$
belongs to the semiparametric model induced by MAR. Under
mild assumptions (e.g., similar to Theorem \ref{thm:asympc}),
an influence function-based estimator using
$\dot{\tau}_{a, \mathrm{MAR}}(O;P_o)$ will be consistent for
$\tau_a^*(P_o)$. However, we will generally incur some bias
because $\tau_a^*(P_o)$ is now not guaranteed to equal
$\tau_a(P_o)$. Specifically, observe that
\begin{align*}
    & \mu_{a, \mathrm{MAR}}(\boldsymbol{L}) \\
    &= \mathbb{E}_{P_o}(Y \mid \boldsymbol{L}, A = a, R + S = 1) \\
    &= \mathbb{E}_{P_o}\left[\mathbb{E}_{P_o}(Y \mid \boldsymbol{L}, 
    A = a, R, S, R + S = 1) \mid \boldsymbol{L}, A = a, 
    R + S = 1\right] \\
    &= \mu_{a, R}(\boldsymbol{L})P_o[R = 1 \mid \boldsymbol{L}, 
    A = a, R + S = 1] + \mu_{a, S}(\boldsymbol{L})P_o[S = 1 
    \mid \boldsymbol{L}, A = a, R + S = 1] \\
    &= \frac{\mu_{a, R}(\boldsymbol{L})
    \gamma_{a}(\boldsymbol{L}) + \mu_{a, S}(1 -
    \gamma_a(\boldsymbol{L}))
    \eta_{a, 0}(\boldsymbol{L})}{\gamma_a(\boldsymbol{L}) +
    (1-\gamma_a(\boldsymbol{L}))\eta_{a, 0}(\boldsymbol{L})}.
\end{align*}
Thus, the bias of the outcome model is given by
\begin{align*} 
    &\mu_{a, \mathrm{MAR}}(\boldsymbol{L}) -
    \mu_a(\boldsymbol{L}) \\
    &= \frac{\mu_{a, R}(\boldsymbol{L})
    \gamma_{a}(\boldsymbol{L}) + \mu_{a, S}(1 -
    \gamma_a(\boldsymbol{L}))
    \eta_{a, 0}(\boldsymbol{L})}{\gamma_a(\boldsymbol{L}) +
    (1-\gamma_a(\boldsymbol{L}))\eta_{a, 0}(\boldsymbol{L})}
    - \mu_{a, R}(\boldsymbol{L})
    \gamma_{a}(\boldsymbol{L}) - \mu_{a, S}(1 -
    \gamma_a(\boldsymbol{L}))
     \\
    &= \gamma_{a}(\boldsymbol{L})\frac{(1 -
    \gamma_a(\boldsymbol{L})) (1 - \eta_{a, 0}(\boldsymbol{L}))}
    {1 - (1 -
    \gamma_a(\boldsymbol{L})) (1 - \eta_{a, 0}(\boldsymbol{L}))}
    (\mu_{a, R}(\boldsymbol{L}) - \mu_{a, S}(\boldsymbol{L})),
\end{align*}
and the overall bias $\tau_a^*(P_o) - \tau_a(P_o)$ is the 
expectation of this quantity.

\end{appendices}
\end{document}